\documentclass[10pt]{article}

\RequirePackage{natbib}
\RequirePackage{hyperref}
\usepackage[utf8]{inputenc} 

\usepackage{amsmath, amssymb, amsthm}
\usepackage{algorithmic}
\usepackage{url}
\usepackage{algorithm}
\usepackage{color}
\newtheorem{theorem}{Theorem}
\newtheorem{lemma}{Lemma}

\usepackage{graphicx}
\usepackage{rotating}
\usepackage[ampersand]{easylist}
\usepackage{morefloats}

\usepackage[table]{xcolor}
\definecolor{lightgray}{gray}{0.9}

\newcommand{\red}{\prec}
\newcommand{\redeq}{\preccurlyeq}
\usepackage{mathtools}
\DeclarePairedDelimiter{\ceil}{\lceil}{\rceil}

\begin{document}

\title{A Unifying Model of Genome Evolution Under Parsimony}
\author{Benedict Paten$^{1*}$, Daniel R. Zerbino$^{2}$, Glenn Hickey$^1$,\\ and David Haussler$^{1,3*}$}
\date{}

\maketitle
\bibliographystyle{abbrvnat}

$^1$ Center for Biomolecular Sciences and Engineering, CBSE/ITI, UC Santa Cruz, 1156 High St, Santa Cruz, CA 95064, USA.\\
$^2$ EMBL-EBI, Wellcome Trust Genome Campus, CB10 1SD Cambridge, UK\\
$^3$ Howard Hughes Medical Institute, University of California, Santa Cruz, CA 95064, USA.\\

$^*$ To whom correspondence should be addressed\\


\section{Abstract}

We present a data structure called a history graph that offers a practical basis for the analysis of genome evolution. It conceptually simplifies the study of parsimonious evolutionary histories by representing both substitutions and double cut and join (DCJ) rearrangements in the presence of duplications. The problem of constructing parsimonious history graphs thus subsumes related maximum parsimony problems in the fields of phylogenetic reconstruction and genome rearrangement.
We show that tractable functions can be used to define upper and lower bounds on the minimum number of substitutions and DCJ rearrangements needed to explain any history graph. These bounds become tight for a special type of unambiguous history graph called an ancestral variation graph (AVG), which constrains in its combinatorial structure the number of operations required.  We finally demonstrate that for a given history graph $G$, a finite set of AVGs describe all parsimonious interpretations of $G$, and this set can be explored with a few sampling moves. 



\section{Introduction}

In genome evolution there are two interacting relationships between nucleotides of DNA resulting from two key features: DNA nucleotides descend from common ancestral nucleotides, and they are covalently linked to other nucleotides.  In this paper we explore the combination of these two relationships in a simple graph model, allowing for change by the process of replication, where a complete sequence of DNA is copied, by substitution, in which the chemical characteristics of a nucleotide are changed, and by the coordinated breaking and rematching of covalent adjacencies between nucleotides in rearrangement operations. These processes have quite different dynamics: DNA molecules replicate essentially continuously, much more rarely substitutions occur and more rarely still rearrangement operations take place. For this reason, and because of inherent complexity issues, a wealth of models, data structures and algorithms have studied these processes either in isolation or in a more limited combination. 

Such evolutionary methods generally start with a set of observed sequences in an alignment, an alignment being a partitioning of elements in the sequences into equivalence classes, each of which represents elements that are homologous, i.e. that share a recognisably recent common ancestor. Though alignments represent an uncertain inference, and though there optimisation for standard models is intractable for multiple sequences (\cite{Elias:2006fe}), we make the common assumption that the alignment is given, as efficient heuristics exist to compute reasonable genome alignments (\cite{Miller:2007gx}, \cite{Darling:2010iu}, \cite{Paten:2011fva}). 

If the sequences in an alignment only differ from one another by substitutions and rearrangements that delete subsequences, or insert novel subsequences (collectively indels), then the alignment data structure is naturally a 2D matrix.  In such a matrix, by convention, the rows represent the sequences and the columns represent the equivalence classes of elements. The sequences are interspersed with ``gap'' symbols to indicate where elements are missing from a column due to indels. From such a matrix alignment, phylogenetic methods infer a history of replication (\cite{Felsenstein:2004ws}). Such a history is representable as a phylogenetic tree, whose internal nodes represent the most recent common ancestors (MRCA) of subsets of the input sequences. To create a history including the MRCA sequences, additional rows can be added to the matrix (\cite{Blanchette:2004bi}, \cite{Kim:2007en}, \cite{Paten:2008bp}). Both the problem of imputing maximum parsimony phylogenetic trees from matrix alignments and calculating maximum parsimony MRCA sequences given a phylogenetic tree and a matrix alignment are NP-hard (\cite{DAY:1987cs}, \cite{Chindelevitch:2006vm}).

In addition to substitutions and short indels, homologous recombination operations are a common modifier of individual genomes within a population. The alignment of long DNA sequences related by these operations is also representable as a matrix. However, the history of replication of such an alignment is no longer generally representable as a single phylogenetic tree, as each column in the matrix may have its own distinct tree. To represent the MRCAs of such an alignment requires a more complex data structure, termed an ancestral recombination graph (ARG) (\cite{Song:2005vj}, \cite{Westesson:2009ja}). It is NP-hard under the infinite sites model (no repeated or overlapping changes) to determine the minimum number of homologous recombinations needed to explain the evolutionary history of a given set of sequences, and probably NP-hard under more general models (\cite{Wang:2001cs}).

Larger DNA sequences, or complete genomes, are often permuted by more complex rearrangements, such that the matrix alignment representation is insufficient. Instead, the alignment naturally forms a graph called a breakpoint graph (\cite{Bergeron:2006uj,Alekseyev:2008vv}). Assuming rearrangements are balanced (neither involving the gain or loss of material), inferring parsimonious rearrangement histories between two genomes has polynomial or better time complexity, whether based upon inversions (\cite{Hannenhalli:1999tb}), translocations (\cite{Bergeron:2006gw}) or double-cut-and-join (DCJ) operations (\cite{Yancopoulos:2005bm}). However, for three or more genomes with balanced rearrangements (\cite{Caprara:1999uj}) or when rearrangements are unbalanced (involving the gain or loss of material) leading to duplications (additional copies of subsequences resulting from rearrangement), these exact parsimony methods are intractable. Exact solutions in the most general case are therefore only feasible for relatively small problems (\cite{Xu:2009ct}) before heuristics become necessary (\cite{Bourque:2002vr, Ma:2008js}). 

Despite the hardness of the general case,  there has been substantial work on computing maximum parsimony results, allowing for a wider repertoire of rearrangements.  El-Mabrouk studied inversions and indels, though gave no exact algorithm for the general case (\cite{ElMabrouk:2000ux}). Recently Yancopoulous (\cite{Yancopoulos:2009ej}) then Braga (\cite{Braga:2011kz}) considered the distance between pairs of genomes differing by DCJ operations and indels, the latter providing the first linear-time algorithm for balanced rearrangements and indels, and the former proposing a data-structure to model duplications. Many methods have been proposed that deal with the combination of rearrangements and duplications, for good recent reviews see (\cite{ElMabrouk:2012bu,Chauve:2013uu}), however until recently there were no algorithms to our knowledge that explicitly unified both duplications and genome rearrangements as forms of general unbalanced rearrangement. First \cite{Bader:2010dd} provided a model allowing for a subset of duplications and deletions as well as balanced DCJ operations, giving a lower bound approximation, while \cite{Shao:2012fn} studied a model allowing atomic (single gene) duplications, insertions and deletions, but arrived at no closed-form formula for the total number of rearrangements. 

The graph model introduced in this paper is capable of representing a general evolutionary history for any combination of replication, substitution and rearrangement  operations, including duplications and homologous recombinations. It therefore generalises phylogenetic trees, graphs representing histories with indels, ancestral recombination graphs and breakpoint graphs, building upon the methods described above. We start by introducing this graph and then develop a maximum parsimony problem that, somewhat imperfectly, generalises maximum parsimony variants of all the problems mentioned, facilitating the study of all these subproblems in one unified domain.  We adopt the common assumption that all substitutions and rearrangements occur independently of one another, and account for tradeoffs between them by independent rearrangement and substitution costs, which are themselves essentially sums over the numbers of inferred events. Importantly, replications that are combined with unbalanced rearrangements are costed by the underlying rearrangement cost. We finally provide a bounded sampling approach to cope with the NP-hardness of the general maximum parsimony problem.


\section{Results}

\subsection{Sequence Graphs and Threads}
\label{threads}

Sequence graphs are used extensively in  comparative genomics, in rearrangement theory typically under the name (multi or master) breakpoint graph (\cite{Bergeron:2006uj}, \cite{Alekseyev:2008vv}, \cite{Ma:2008js}) and in alignment under the name A-bruijn (\cite{Raphael:2004hj}) or adjacency graph (\cite{Paten:2011fv}). We use the following bidirected form, which is similar to that used by \cite{Medvedev:2009fb} for sequence assembly. 

A \emph{(bidirected) sequence graph} $G = (V_G, E_G)$ is a graph in which a set $V_G$ of vertices are connected by a set $E_G$ of bidirected edges (\cite{Edmonds:2003jq}), termed \emph{adjacencies}.
A vertex represents a subsequence of DNA termed a \emph{segment}.
A vertex $x$ is oriented, having a \emph{tail side} and a \emph{head side}, respectively denoted $x_{head}$ and $x_{tail}$. These categories $\{ head, tail \}$ are called \emph{orientations}.
An adjacency, which represents the covalent bond between adjacent nucleotides of DNA, is a pair set of sides.  We refer to the two sides contained in an adjacency as its \emph{endpoints}.
Adjacencies are bidirected, in that each endpoint is not just a vertex, but a vertex with an independent orientation (either head or tail). 
For convenience, we say a side is \emph{attached} if it is contained in an adjacency, else it is \emph{unattached}.
By extension, we say a vertex is \emph{attached} if either of its sides are attached, else it is \emph{unattached}.

Associated with a sequence graph is a \emph{labeling}, 
\emph{i.e.} a function $l : V_G \rightarrow \Sigma^* \mbox{ }\cup \mbox{ } \{ \emptyset \}$ where $\Sigma = \{ A/T, C/G, G/C, T/A \}$ is the alphabet of \emph{bases}, which are oriented, paired nucleotides of DNA, and $\Sigma^*$ is the set of all possible \emph{labels} consisting of finite sequences of bases in $\Sigma$. 
Bases and labels are directed. For $\rho/\tau \in \Sigma$, $\rho$ is the \emph{forward complement} and $\tau$ is the \emph{reverse complement}. 
If a vertex is traversed from its tail to its head side, its label is read as the sequence of its forward complements. Conversely, if traversed from head to tail, the label is read as the reverse sequence of the reverse complements. A vertex $x \in V_G$ for which $l(x) = \emptyset$ is \emph{unlabeled}. A label represents a multibase allele. 

A \emph{thread} is a connected component in a sequence graph in which each side is connected to at most one adjacency. A \emph{thread graph} is a sequence graph in which every connected component is a thread. 
In this paper we limit ourselves to investigating thread graphs. 
A thread may be a simple cycle, representing a circular DNA molecule, or have two unattached sides, in which case it represents a linear DNA molecule or fragment of a larger DNA molecule.  An example thread graph is shown in Figure \ref{threadGraphs}.
A thread graph is phased, in that each thread is assigned a maximal DNA sequence (and its reverse complement), and any path though that thread corresponds to a subsequence of these maximal sequences. In contrast, a sequence graph that is not a thread graph may be unphased, in that there exist many possible maximal sequences for each of its connected components.

\begin{figure}[h!]
\begin{center}
\includegraphics[width=13cm]{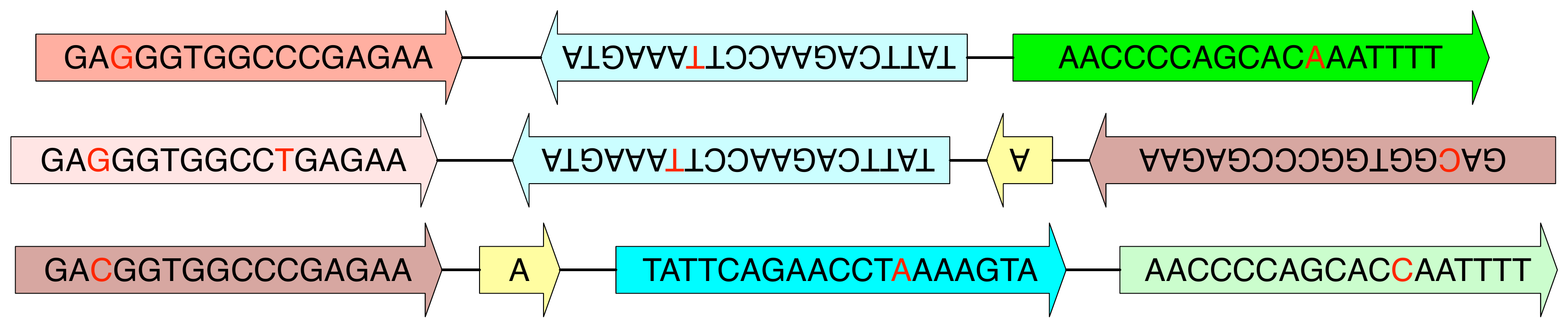}
\caption{A thread graph. For visual appeal, vertices are the arrow shapes with the sides indicated by the ends of the arrows. Labels within the arrows represent the subsequence of DNA when traversed from the tail to the head side of the arrow, and are read as the reverse complement when traversed from the head to the tail side. Adjacencies are the lines connecting the ends of the arrow shapes.  They are bidirected, i.e. there are 3 unordered types: head-tail (symmetrically tail-head), tail-tail and head-head adjacencies. In prior illustrations of bidirected graphs (\cite{Medvedev:2009fb}) orientations were drawn on the lines, however the semantics of the graph are still the same, in that head and tail orientations are properties of the endpoints of the adjacencies, not the vertices.
The graph contains three linear threads. As an example, because the middle vertex is attached in the opposite direction and therefore reverse-complemented when traversed left-to-right, the top thread represents the sequence 
``GAGGGTGGCCCGAGAA TACTTTAAGGTTCTGAATA AACCCCAGCACAAATTTT''
(from left-to-right, spaces used to distinguish vertex labels) and its reverse complement, 
``AAAATTTGTGCTGGGGTT TATTCAGAACCTTAAAGTA TTCTCGGGCCACCCTC" (from right-to-left). The colours (red, blue, green and yellow) of the arrows represent homologies between the vertices, these are not part of the thread graph itself, but are used in subsequent figures that build on this example. Different hues of a colour and the red letters represent differences between labels of the same colour.
}
\label{threadGraphs}
\end{center}
\end{figure}

\subsection{History Graphs}

Nucleotides of DNA derive from one another by a process of replication. This replication process is represented in history graphs, which add ancestry relationships to thread graphs.

A \emph{history graph} $G = (V_G, E_G, B_G)$ is a thread graph with an additional set $B_G$ of directed edges between vertices, termed \emph{branches}. 
Each vertex is incident with at most one incoming branch. 
The \emph{event graph} $D(G)$ is the directed graph formed by the contraction\footnote{The contraction of an edge $e$ is the removal of $e$ from the graph and merger of the vertices $x$ and $y$ incident with $e$ to create new vertex $z$, such that edges incident with $z$ were incident either with $x$ or $y$ or both, in the latter case becoming a loop edge on $z$.}  of adjacencies in $E_G$.
For $G$ to be a history graph $D(G)$ must be a directed acyclic graph (DAG), a property we term \emph{acyclicity}. 
Example history graphs are shown in Figure \ref{dnaHistoryGraphs}(A,B), along with an event graph in \ref{dnaHistoryGraphs}(C) for the history graph shown in \ref{dnaHistoryGraphs}(B). 

To avoid confusion we define terminology to discuss branch relationships.
Each weakly connected component of branches forms a \emph{branch-tree}.
Two vertices are \emph{homologous} if they are in the same branch-tree.
A vertex $y$ is a \emph{descendant} of a vertex $x$, and conversely $y$ is an \emph{ancestor} of $x$, if $y$ is reachable by a directed path of branches from $x$.
If two homologous vertices do not have an ancestor/descendant relationship then they are \emph{indirectly related}.
For a branch $e=(x, y)$, $x$ is the \emph{parent} of $e$ and $y$, and $y$ is the \emph{child} of $e$ and a \emph{child} of $x$.
Similarly, $e$ is the \emph{parent branch} of $y$ and a \emph{child branch} of $x$.
A vertex is a \emph{leaf} if it has no incident outgoing branches, a \emph{root} if it has no incident incoming branches, else it is \emph{internal}. 
We reuse the terminology of parent, child, homologous, ancestor, descendant and indirectly related with sides. Two sides have a given relationship if their vertices have the relationship and they have the same orientation. Similarly, a side is a leaf (resp. root) if its vertex is a leaf (resp. root).

\begin{figure}[h!]
\begin{center}
\includegraphics[width=12cm]{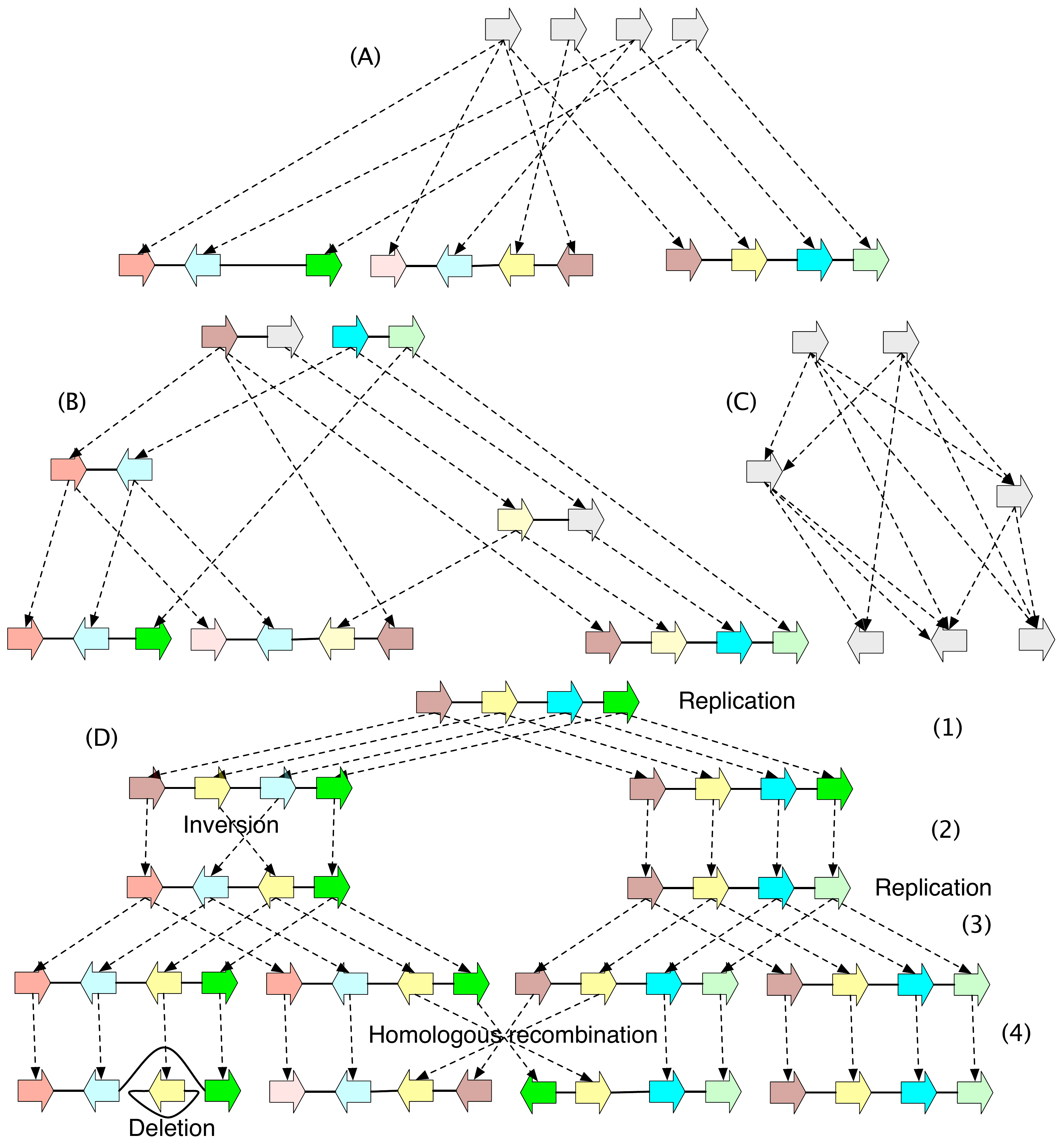}
\caption{\textbf{(A)} A history graph representing homology relationships between the vertices in Figure \ref{threadGraphs}. Due to space, colours are used as labels (and match those in Figure \ref{threadGraphs}), with unlabeled vertices shaded grey. Two vertices have the same colour shade if they have identical labels. The dotted arrows represent branches. Four ancestral vertices are added relative to Figure \ref{threadGraphs} to represent the common ancestral vertices of the subsets of homologous vertices in Figure \ref{threadGraphs}. \textbf{(B)} An extension of (A). \textbf{(C)} The event graph for (B).  \textbf{(D)} A simple history with four epochs (1 - 4), and rearrangements given names corresponding to their type. It is a realisation for the graphs in (A) and (B). 
}
\label{dnaHistoryGraphs}
\end{center}
\end{figure}

\subsection{Simple Histories}

We formally define a class of history graphs, called \emph{simple histories}, for which parsimonious sequences of substitutions and rearrangements can be trivially derived.


A \emph{bilayered history graph} is a history graph whose threads can be partitioned into \emph{root} and \emph{leaf} layers, such that every branch connects a vertex in the root layer with a vertex in the leaf layer. A \emph{rearrangement epoch} is a bilayered history graph in which every branch tree is a root with 1 child, every vertex is labeled, and any set of homologous sides are either all attached or all unattached. For $n \geq 2$, an \emph{$n$-way replication epoch} is a bilayered history graph in which every branch tree is a root with $n$ children, every vertex is labeled, any set of homologous sides are either all attached or all unattached, if two root sides $x_\alpha$ and $y_\beta$ are attached by an adjacency then each child of $x_\alpha$ is attached to a child of $y_\beta$, and a root vertex has at most one child with a label different from its own. An \emph{epoch} is either a rearrangement epoch or an $n$-way replication epoch for some $n \geq 2$. 
A \emph{layered history graph} is a history graph that can be edge partitioned into a finite sequence of bilayered history graphs, such that the leaf layer of a contained bilayered history graph is the root layer of the following bilayered history graph.
A \emph{simple history} is a layered history graph whose bilayered subgraphs are all epochs. 
An example simple history with epoch subgraphs is shown in Figure \ref{dnaHistoryGraphs}(D). 

A \emph{substitution} occurs on a branch if the labels of its endpoints are not identical. Note that a substitution can occur either in a rearrangement or a replication epoch. The \emph{substitution cost} of a simple history $\textbf{H}$ is the total number of substitutions, denoted $s(\textbf{H})$. 
The example simple history in Figure \ref{dnaHistoryGraphs}(D) has substitution cost 4. Note the requirement that all homologous sides in a simple history be either all attached or all unattached does not forbid rearrangements involving the observed ends of chromosomes (linear threads), because it is always possible to add material to a simple history at zero cost that attaches such unattached sides and allows them to participate in rearrangements.

The substitution cost defined deals, abstractly, with changes of alleles in which any change between alleles is scored equally. However for the case $\Sigma^* = \Sigma$, i.e. single base labels, the substitution cost is the minimum number of single base changes. Furthermore, any history graph in which all homologous labels have the same length can easily be converted to a semantically equivalent history graph for which $\Sigma^* = \Sigma$. More complex substitution costs to deal with the case where the alphabet represents the alleles of genes, as is commonly dealt with in rearrangement theory, are straightforward but not pursued here for simplicity.


A \emph{rearrangement cycle} in a rearrangement epoch is a circular path consisting of one or more repetitions of the basic pattern consisting of an adjacency edge in the root layer, a forward branch to the leaf layer, an adjacency edge in the leaf layer and a reverse branch to the root layer. Its \emph{size} is the number of repetitions in it of this basic pattern minus 1. A linear path that follows this same basic pattern but does not complete every pattern and return to the original vertex is a \emph{degenerate rearrangement cycle}. Its size is the size of the smallest rearrangement cycle that can be obtained from it by adding edges.  The \emph{rearrangement cost} of a simple history $\textbf{H}$ is the total size of all rearrangement cycles in it, denoted $r(\textbf{H})$. This cost is known to be the number of double-cut-and-join (DCJ) operations needed to achieve all the rearrangements.

\begin{lemma}
The rearrangement cost of an epoch is the minimum number of double-cut-and-join (DCJ) operations required to convert the root layer's adjacencies into the leaf layer's adjacencies.
\end{lemma}

\begin{proof}
Similar to that given in \cite{Yancopoulos:2005bm}.
\end{proof}

The example simple history in Figure \ref{dnaHistoryGraphs}(D) has rearrangement cost 3.

Because different studies lay different emphases on substitution or rearrangement (e.g. because of the available data) and because the events do not have the same probability in practice, we allow for a degree of freedom in the definition of the overall cost function.
A \emph{(simple history) cost function} for a simple history is any monotone function on the substitution and rearrangement costs in which both substitutions and rearrangements have non-zero cost.

\subsection{Reduction}

Not all history graphs are as detailed as simple histories. We define below a partial order relationship that describes how one graph can be a generalization of another graph, so for example, a less detailed history graph can be used to subsume multiple simple histories. 




A branch whose child is unlabeled and unattached is referred to as having a \emph{free-child}. 
A branch whose parent is unlabeled, unattached and a root with a single child is referred to as having a \emph{free-parent}. A vertex is \emph{isolated} if it has no incident adjacencies or branches.
A \emph{reduction operation} is an operation upon a history graph that either:
\begin{itemize}
\item Deletes an adjacency, an isolated vertex or the label of a vertex.
\item Contracts a branch with a free-child or free-parent.
\end{itemize}
See Figure \ref{reductionOperations}(A-E) for examples.
The inverse of a reduction operation is an \emph{extension operation}.

\begin{figure}[h!]
\begin{center}
\includegraphics[width=8cm]{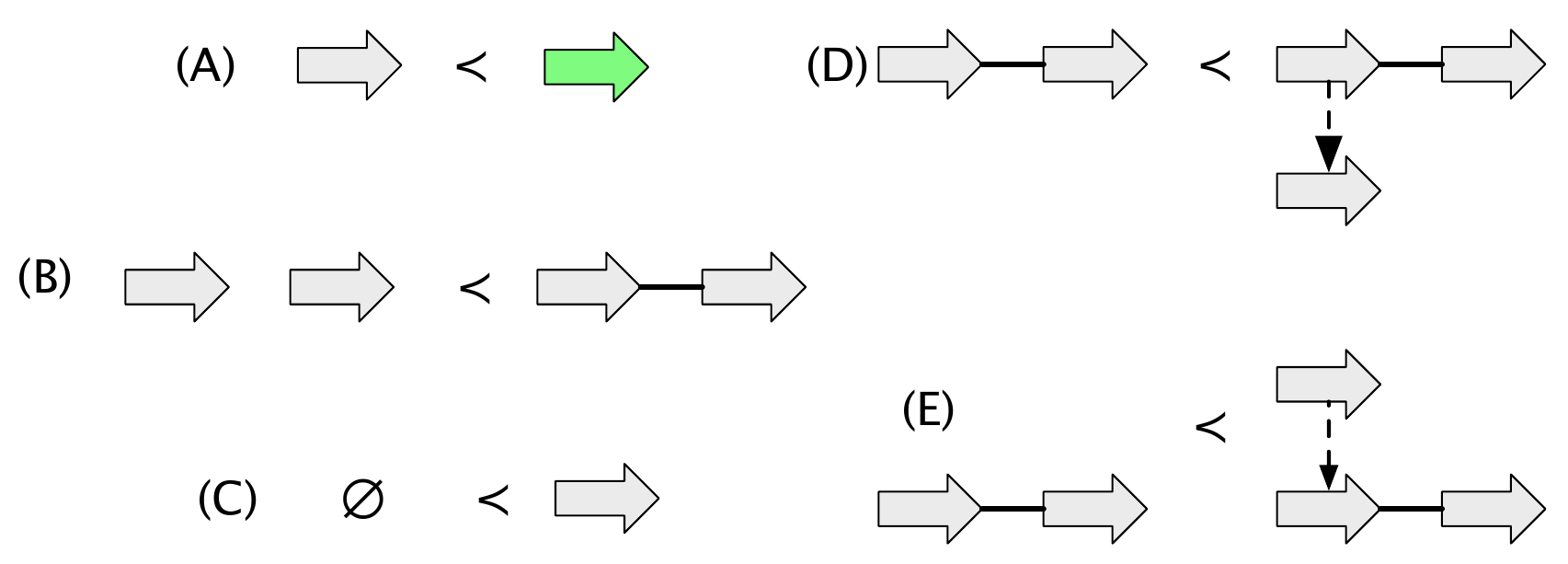}
\caption{\textbf{(A-E)} Reduction operations. For each case the graph on the left is a reduction of the graph on the right.  \textbf{(A)} A label deletion. \textbf{(B)} An adjacency deletion. \textbf{(C)} A vertex deletion. \textbf{(D)} A contraction of a branch with a free-child. \textbf{(E)} A contraction of a branch with a free-parent.}
\label{reductionOperations}
\end{center}
\end{figure}

\begin{lemma}
The result of a reduction operation is itself a history graph.
\end{lemma}


A history graph $G$ is a \emph{reduction} of another history graph $G'$ if $G$ is isomorphic to a graph that can be obtained from $G'$ by a sequence of reduction operations, termed a \emph{reduction sequence}.

\begin{lemma}
The reduction relation is a partial order.
\end{lemma}


We write $G \redeq G'$ to indicate that $G$ is a reduction of $G'$ and $G \red G'$ to indicate that $G$ is a reduction of $G'$ not equal to $G'$.
Like reduction and extension operations, if $G$ is a reduction of $G'$, $G'$ is an \emph{extension} of $G$.
An examination of the reduction relation is in the discussion section and Figure \ref{justifyingReduction}.

\subsection{History Graph Cost}

Using the parsimony principle, we now extend parsimony cost functions, previously defined on simple histories, to all history graphs.

A simple history $\textbf{H}$ that is an extension of a history graph $G$ is called a \emph{realisation} of $G$. The set $\mathcal{H}(G)$ is the realisations of G.
For a given cost function $c$ the \emph{cost} of a history graph $G$ is\footnote{Note: while $\mathcal{H}(G)$ is infinite we show in the sequel that the infimum of this set of costs is always achieved by a history, hence the infimum is the minimum.} \[ C(G, c) = \underset{\textbf{H} \in \mathcal{H}(G)}{\min\mbox{ }}  c(s(\textbf{H}), r(\textbf{H})). \] 


\begin{lemma}
The problem of finding the cost of a history graph is NP-hard.
\end{lemma}

\begin{proof}
There are parsimony problems on either substitutions or rearrangements alone that are NP-hard and can be formulated as special cases of the problem of finding the minimum cost realisation of a history graph (\cite{DAY:1987cs}, \cite{Tannier:2009jp}). 
\end{proof}

\subsection{The Lifted Graph}

Although determining the cost of a history graph is NP-hard, we will show that the cost can be bounded such that the bounds become tight for a broad, characteristic subset of history graphs.
To do this we introduce the concept of lifted labels and adjacencies, which are used to project information about labels and adjacencies from descendant to ancestral vertices and are useful in reasoning about the cost of a history graph.

The \emph{free-roots} of a history graph $G$ are a set of additional vertices such that a single, unique free-root is assigned to each root vertex in $G$ (see the top of Figure \ref{liftedGraph}(A)).
For a vertex $x$, its \emph{lifting ancestor} $A(x)$ is the most recent labeled ancestor of $x$, else if no such vertex exists, the free-root of the branch-tree containing $x$.
For a side $x_{\alpha}$ its \emph{lifting ancestor} (overloading notation) $A(x_{\alpha})$ is its most recent attached ancestor, else if no attached ancestor exists, its ancestral side in the free root of the branch tree containing it. 

For a labeled vertex $y$, a \emph{lifted label} is a label identical to $l(y)$ on its lifting ancestor. For a vertex the \emph{lifted labels} is therefore a multiset, because the same lifted label may be lifted to a lifting ancestor from multiple distinct descendants and each is considered an element of the multiset. 


For an adjacency $\{ x_\alpha, y_\beta \}$, a \emph{lifted adjacency} is a bidirected edge $\{ A(x_\alpha), A(y_\beta) \}$.
In analogy with the lifted labels for a vertex, the \emph{lifted adjacencies} for a side is the multiset of lifted adjacencies incident with the side. 

A history graph $G$ with free-roots, lifted labels and lifted adjacencies is a \emph{lifted graph} $\textbf{L}(G)$.
Figure \ref{liftedGraph}(A) shows an example lifted graph that outlines these concepts. 


\begin{figure}[h!]
\begin{center}
\includegraphics[width=12cm]{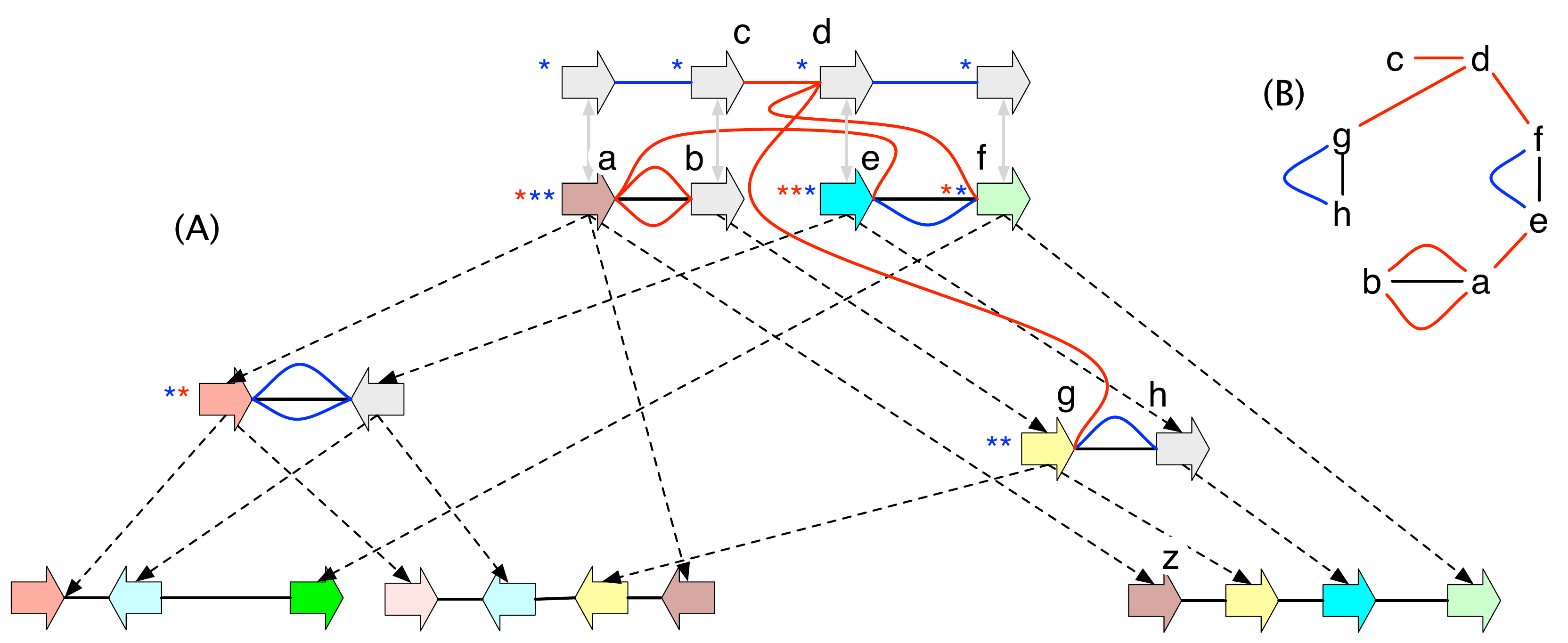}
\caption{\textbf{(A)} The lifted graph for the history graph in Figure \ref{dnaHistoryGraphs}(B). 
The blue and red lines represent, respectively, trivial and non-trivial lifted adjacencies. Similarly, the blue and red stars represent, respectively, trivial and non-trivial lifted labels. The free-roots are shown as a set of vertices above the other vertices, with a grey line identifying their matching branch-tree. \textbf{(B)} The module in (A) containing non-trivial lifted edges. Lower case letters are used to identify the sides. }
\label{liftedGraph}
\end{center}
\end{figure}


Some lifted elements do not imply change between descendant and ancestral states, while others do. To formalise such a notion we define trivial and non-trivial labels and and adjacencies. 
A lifted label $\rho$ of a labeled vertex $x$ is \emph{trivial} if $l(x) = \rho$. A lifted label $\rho$ on an unlabeled vertex $x$ (necessarily a free root) is trivial if it is the only lifted label on $x$. Otherwise a lifted label is \emph{non-trivial}.

A \emph{junction side} is a most recent common ancestor (MRCA) of two attached, indirectly related sides.
For a history graph $G$, a lifted adjacency $e = \{ A(x_{\alpha}), A(y_{\beta}) \}$ is \emph{trivial} if there exists no unattached junction side on the path of branches from (but excluding) $A(x_{\alpha})$ to (but excluding) $x_{\alpha}$, or on the path of branches from (but excluding) $A(y_{\beta})$ to (but excluding) $y_{\beta}$ and either there is a (regular) adjacency between $A(x_{\alpha})$ and $A(y_{\beta})$ in $G$ or $A(x_{\alpha})$ and $A(y_{\beta})$ are free roots, else $e$ is \emph{non-trivial}.
See Figure \ref{liftedGraph}(A) for examples of trivial and non-trivial labels and adjacencies.


\subsection{Ancestral Variation Graphs}

We can now define a broad class of history graphs for which cost can be computed in polynomial time. 
To do this we will define ambiguity, information that is needed to allow the tractable assessment of cost. There are two types of ambiguity.

The \emph{substitution ambiguity} of a history graph $G$,  denoted $u_s(G)$, is the total number of non-trivial lifted labels in excess of one per vertex.
Substitution ambiguity reflects uncertainty about MRCA bases. The substitution ambiguity of the history graph in Figure \ref{dnaHistoryGraphs}(B) is 1, as there exists one vertex with two non-trivial lifted labels. 

The \emph{rearrangement ambiguity} of a history graph $G$, denoted $u_r(G)$, is the total number of non-trivial lifted adjacency incidences in excess of one per side.
Rearrangement ambiguity reflects uncertainty about MRCA adjacencies. The rearrangement ambiguity of the history graph in Figure \ref{dnaHistoryGraphs}(B) is 5, because two sides have three incident non-trivial lifted edges and one side has two incident non-trivial lifted edges.

The \emph{ambiguity} of a history graph $G$ is $u(G) = u_s(G) + u_r(G)$. An \emph{ancestral variation graph} (AVG) $H$ is a history graph such that $u(H) = 0$, i.e. an unambiguous history graph.

\begin{lemma}
Simple histories are AVGs.
\end{lemma}


While simple histories are AVGs, so are many other history graphs that are far less detailed. For example, the AVG in Figure \ref{AVGExample} is not a simple history.

\begin{figure}[h!]
\begin{center}
\includegraphics[width=12cm]{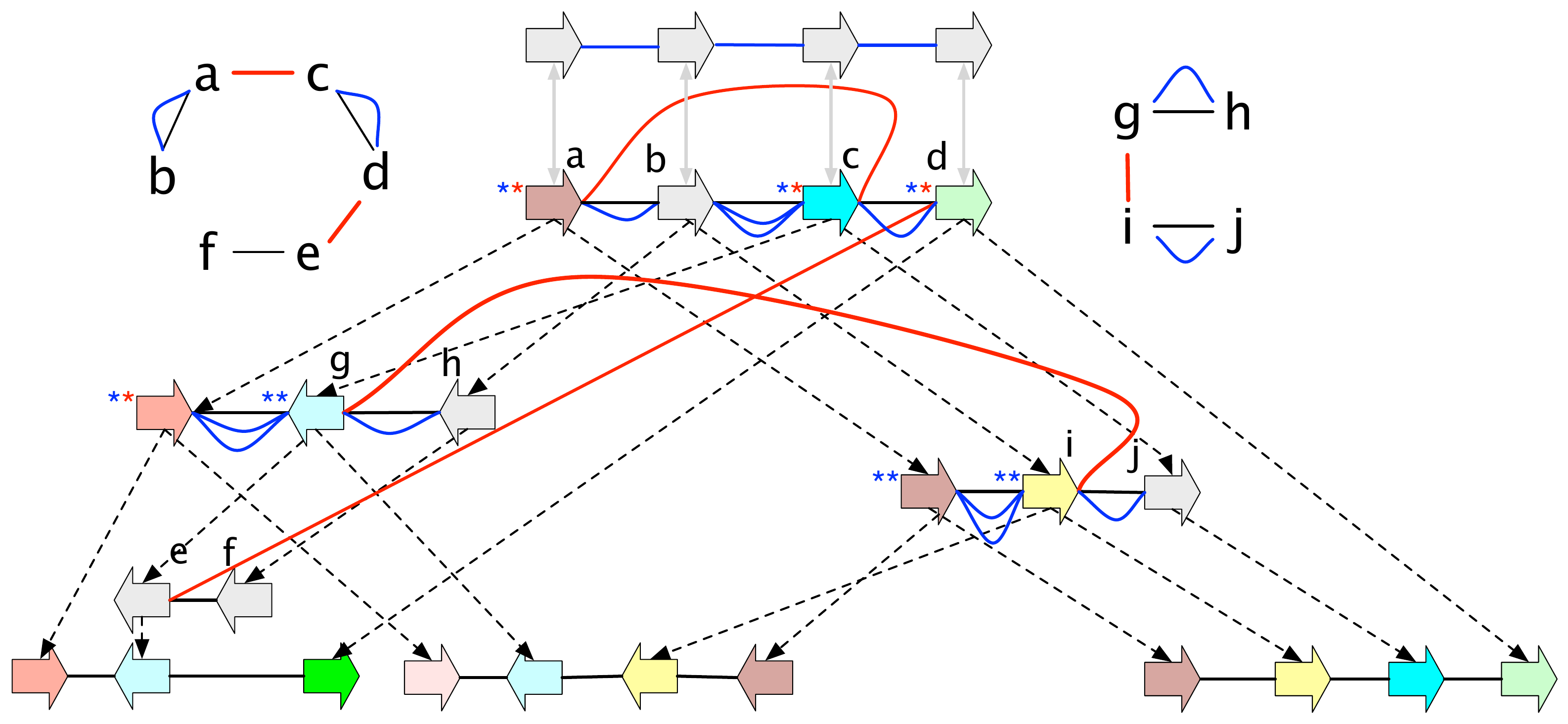}
\caption{The lifted graph for an AVG with (simple) modules containing non-trivial lifted adjacencies highlighted, using the same notation as in Figure \ref{liftedGraph}(A).}
\label{AVGExample}
\end{center}
\end{figure}

\
\subsection{Bounds on Cost}

We provide trivially computable lower and upper bound cost functions for history graphs that are tight for AVGs.

The \emph{lower bound substitution cost} (LBSC) of a history graph $G$, denoted $s_l(G)$,  is the total number of distinct (not counting duplicates in the multiset) nontrivial lifted labels at all vertices minus the number of unlabeled vertices with non-trivial lifted labels (necessarily free roots). The LBSC of the history graph in Figure \ref{dnaHistoryGraphs}(B) is 4.

The \emph{upper bound substitution cost} (UBSC) of a history graph $G$, denoted $s_u(G)$, is the total number of nontrivial lifted labels at all vertices minus the maximum number of identical lifted labels at each unlabeled vertex with non-trivial lifted labels (again, necessarily free roots).
The UBSC of the history graph in Figure \ref{dnaHistoryGraphs}(B) is 5. For the AVG in Figure \ref{AVGExample}, LBSC = UBSC = 4.

The \emph{module graph} of a history graph $G$ is a multi-graph in which the vertices are the sides of vertices in $L(G)$ that have incident real or lifted adjacencies and the edges are the real and lifted adjacencies in $L(G)$ incident with these sides.
Each connected component in a module graph is called a \emph{module}. 
The set of modules in the module graph for $G$ is denoted $M(G)$.
Figure \ref{liftedGraph}(B) shows the modules for Figure \ref{liftedGraph}(A). 

The \emph{lower bound rearrangement cost} (LBRC) for a history graph $G$ is: \[ r_l(G) = \sum_{M \in M(G)} (\ceil[\Big]{\frac{|V_M |}{2}} - 1).\]

For a history graph that is a simple history this definition is equivalent to the earlier definition of rearrangement cost for simple histories. 


The \emph{upper bound rearrangement cost} (UBRC) of a history graph $G$, denoted $r_u(G)$, is the total number of non-trivial lifted adjacencies in $L(G)$ minus the number of modules in $M(G)$ in which every side has exactly one incident non-trivial lifted edge.
The LBRC of the history graph in Figure \ref{dnaHistoryGraphs}(B) is 3 and its UBRC is 6. For the AVG in Figure \ref{AVGExample} LBRC = UBRC = 3.

\begin{theorem}
\label{avgTheorem}
For any history graph $G$ and any cost function $c$, $c(s_l(G), r_l(G)) \le C(G, c) \le c(s_u(G), r_u(G))$ with equality if $G$ is an AVG.
\end{theorem}

The proof is given in Appendix A.

Theorem \ref{avgTheorem} demonstrates that LBSC and LBRC are lower bounds on cost, UBSC and UCRC are upper bounds on cost, and that all these bounds become tight at the point of zero ambiguity. This implies that to assess cost of an arbitrary history graph $G$ we need only search for extensions of $G$ to the point that they have zero ambiguity and not the complete set of simple history realisations of $G$. For an AVG $H$, as the lower and upper bounds on cost are equivalent, we write $r(H) = r_l(H) = r_u(H)$ and $s(H) = s_l(H) = s_u(H)$.

\subsection{$G$-Optimal AVGs}

We now explore the process of sampling AVG extensions of an initial starting graph.
Though it is possible to start from any history graph, in practice we are likely to start from a history graph $G$ based on sequence alignments, such as that shown in Figure \ref{dnaHistoryGraphs}(A). If $G$ is already an AVG, by Theorem \ref{avgTheorem}, it is trivial to assess its cost. If not we sample AVG extensions of $G$ in order to assess cost and explore the set of most parsimonious realisations of $G$. With the aim of restricting this search, ultimately to a finite space, we first define the set of \emph{$G$-optimal AVGs}.

An AVG extension $H$ of a history graph $G$  is \emph{$G$-parsimonious} w.r.t. a cost function $c$ if  $C(G,c) = c(s(H), r(H))$.
The set of $G$-parsimonious AVGs is necessarily infinite: it is always possible to add arbitrary vertices without affecting substitution or rearrangement costs. To avoid the redundant sampling of AVG extensions of $G$ and their own extensions we define the notion of minimality.

An AVG extension $H$ of $G$ is \emph{$G$-minimal} if there is no other AVG $H'$ such that $G \red H' \red H$.
The set of $G$-minimal AVGs contains those AVGs that can not be reduced without either ceasing to be AVGs or extensions of $G$. 
This set is also infinite for some DNA history graphs (Lemma \ref{infiniteMinimalExtensions} below). 

An AVG is \emph{$G$-optimal} w.r.t. a cost function $c$ if it is both $G$-parsimonious w.r.t. to $c$ and $G$-minimal.
We establish below that the set of $G$-optimal AVGs is finite for any history graph $G$. 
By definition, any $G$-parsimonious AVG is either $G$-minimal or has a $G$-minimal reduction therefore we can implicitly represent and explore the set of parsimonious realisations of $G$ by sampling just the $G$-optimal AVGs. 
 
 \subsection{$G$-Bounded History Graphs}

Unfortunately, because the history graph cost problem is NP-hard, it is unlikely that there exists an efficient way to sample only $G$-optimal. Instead, we now define a finite bounding set that contains $G$-optimal and can be efficiently searched. Conveniently this bounding set is the same for all cost functions.

A label of a vertex $x$ is a \emph{junction} (overloading the term junction, but using it analogously) if $x$ has more than one lifted label, else it is a \emph{bridge} if $x$ has one lifted label, its lifted label is non-trivial, the most recent labeled ancestor of $x$ is labeled the same as $x$ and this ancestor has at least one non-trivial lifted label (see Figure \ref{nonMinimalElements}(A,B)).


A side $x_\alpha$ is a \emph{bridge side} if it is not a junction, is incident with one non-trivial lifted adjacency and an adjacency $e$ that defines a trivial lifted adjacency $e'$ whose $A(x_{\alpha})$ endpoint is a junction side incident with a non-trivial lifted adjacency, and such that if $e$ is deleted at least one endpoint of $e'$ in the original graph remains a junction side in the resulting graph (see Figure \ref{nonMinimalElements}(C,D)). An adjacency is a \emph{junction} (again, overloading the term junction) if either of its endpoints are junctions, else it is a \emph{bridge} (overloading bridge) if either of its endpoints are bridge sides.

An element is \emph{non-minimal} if it is a branch with a free-child or free-parent, an isolated vertex, or label or adjacency that is not a junction or bridge.

\begin{figure}[h!]
\begin{center}
\includegraphics[width=12cm]{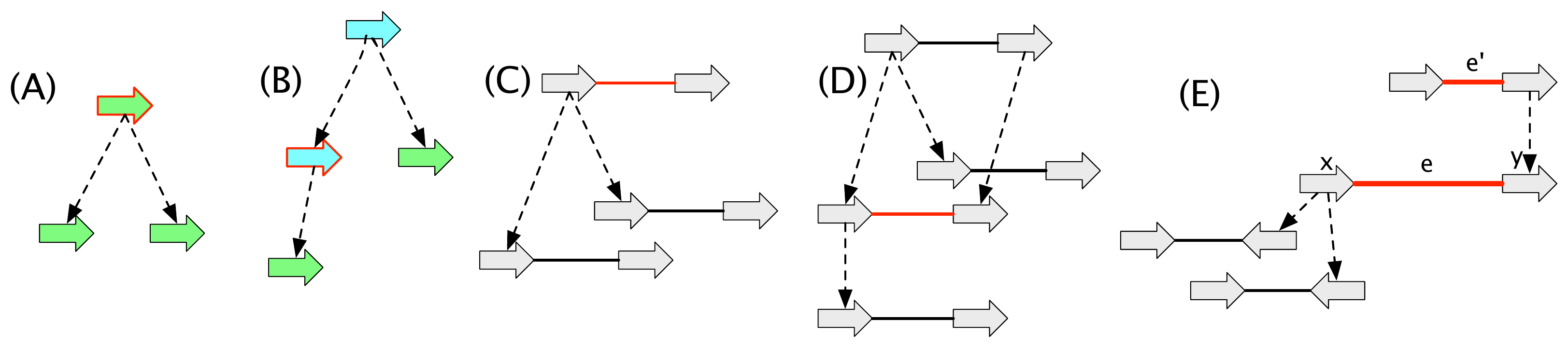}
\caption{\textbf{(A)} A junction label. \textbf{(B)} A bridge label. \textbf{(C)} A junction adjacency. \textbf{(D)} A bridge adjacency. \textbf{(E)} An example of a pair of ping-pong adjacencies. The named elements are outlined in red. } 
\label{nonMinimalElements}
\end{center}
\end{figure}

For $G \redeq G'$, an element in $G'$ is \emph{$G$-reducible} if there exists a reduction operation in a reduction sequence from $G'$ to $G$ that either deletes the element if it is an adjacency, label or vertex or contracts it if it is a branch. 
We are interested in the set of $G$-reducible elements of an extension of $G$, as they are the elements which may be added and removed during an iterative sampling procedure. 

For $G \redeq G'$, the \emph{$G$-unbridged graph} of $G'$ is the reduction resulting from the deletion of all $G$-reducible bridge adjacencies in $G'$.
A side $x_\alpha$ that has no attached descendants is a \emph{hanging side}. 
A pair of adjacencies $e$ and $e'$, each with a hanging side, and such that $e$ has an endpoint whose most recent attached ancestor is incident with $e'$, form a pair of \emph{ping-pong adjacencies}. We call $e$ the \emph{ping adjacency} and $e'$ the \emph{pong adjacency} (Figure \ref{nonMinimalElements}(E)).

A history graph $G'$ is \emph{$G$-bounded} if it is an extension of $G$ that does not contain a $G$-reducible non-minimal element and its $G$-unbridged graph does not contain a $G$-reducible ping adjacency.  

\begin{theorem}
\label{containsGOptimal}
The set of $G$-bounded AVGs contains the $G$-optimal AVGs for every cost function.
\end{theorem}

The proof is given in Appendix B.

Importantly, the following theorem demonstrates that there is a constant $k$ such that any $G$-bounded history graph is at most $k$ times the cardinality of $G$.

\begin{theorem} 
\label{finiteGBoundedTheorem}
A $G$-bounded history graph contains less than or equal to $\max(0, 10n - 8)$ $G$-reducible adjacencies and $\max(0, 2m - 2, 20n - 16, 20n + 2m - 18)$ additional vertices, where $n$ is the number of adjacencies in $G$ and $m$ is the number of labeled vertices in $G$. This bound is tight for all values of $n$ and $m$.
\end{theorem}

The proof is given in Appendix C.

The set of $G$-bounded history graphs and, by inclusion, the set of $G$-optimal AVGs are therefore finite.

\subsection{The $G$-bounded Poset}

Finally we demonstrate how to navigate between $G$-bounded history graphs using a characteristic set of operations that define a hierarchy between these graphs.

\begin{figure}[h!]
\begin{center}
\includegraphics[width=12cm]{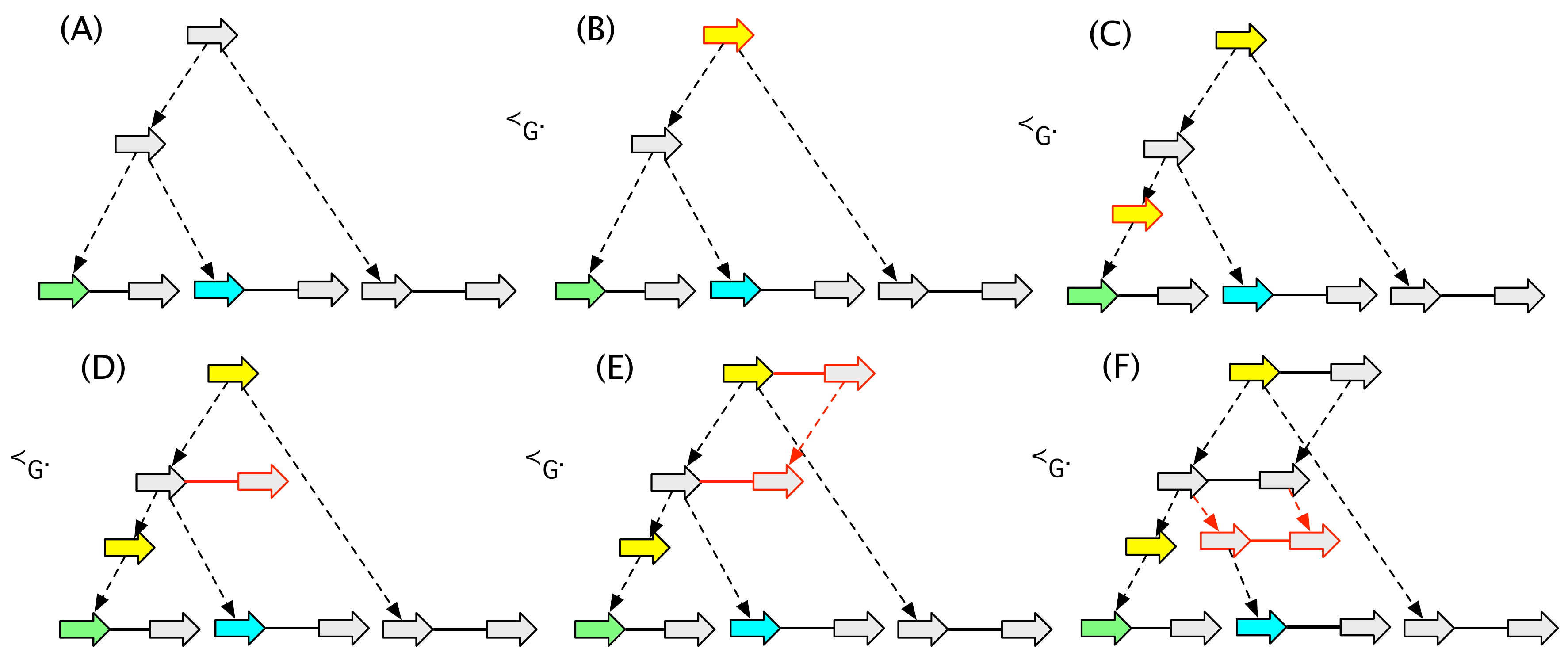}
\caption{A sequence of $G$-bounded extension operations that convert the graph in \textbf{(A)} into the AVG in \textbf{(F)}. 
}
\label{walkingForwardsExample}
\end{center}
\end{figure}

For a vertex $x$ in a $G$-bounded history graph the \emph{composite minimisation} of $x$ is as follows:
\begin{itemize}
\item If $x$ is unattached and unlabeled and has a $G$-reducible parent branch, the contraction of the parent branch, renaming the resulting merged vertex $x$.
\item If $x$ is then an unattached, unlabeled root and has a single $G$-reducible child branch, the contraction of the child branch, renaming the resulting merged vertex $x$.
\item The deletion of $x$ if subsequently isolated, unlabeled and $G$-reducible.
\end{itemize}
A \emph{$G$-bounded reduction operation} on a $G$-bounded history graph is one of the following operations, provided it results in a $G$-bounded history graph.
\begin{itemize}
\item  a \emph{label detachment}: the deletion of a $G$-reducible label on a vertex $x$, followed by the composite minimisation of $x$ (Figure \ref{walkingForwardsExample}(A-C)).
\item an \emph{adjacency detachment}: the deletion of a $G$-reducible adjacency $\{ x_{\alpha}, y_{\beta} \}$ followed by the composite minimisation of $x$ and $y$ (Figure \ref{walkingForwardsExample}(D-F)). The inverse of an adjacency detachment is an \emph{adjacency attachment}.
\item a \emph{lateral-adjacency detachment}: the adjacency detachment of a pair of $G$-reducible junction adjacencies $\{ x_{\alpha}, y_{\beta} \}$ and $\{ A(x_{\alpha}), A(y_{\beta}) \}$, and a subsequent adjacency attachment that creates an adjacency that includes $x_{\alpha}$ or $y_{\beta}$ as an endpoint (Figure \ref{walkingForwardsExample}(D-E)).
\end{itemize}

Note that the first two $G$-bounded reduction operations are combinations of reduction operations, while the lateral-adjacency detachment, which proves necessary to avoid creating intermediate graphs with $G$-reducible ping-pong edges, involves both reduction and extension operations, but always reduces the total number of adjacencies. 
As with reduction operations, the inverse of a $G$-bounded reduction operation is a \emph{$G$-bounded extension operation}.
A $G$-bounded history graph $G'$ is a \emph{$G$-bounded reduction} (resp. extension) of another $G$-bounded history graph $G''$ if $G'$ is isomorphic to a graph that can be obtained from $G''$ by a sequence of $G$-bounded reduction (resp. extension) operations.

\begin{lemma}
The $G$-bounded reduction relation is a partial order. 
\end{lemma}


The \emph{$G$-bounded poset} is the set of $G$-bounded history graphs with the $G$-bounded reduction relation. We write $\red_G$ to denote the $G$-bounded reduction relation and $\red\cdot_G$ to denote its covering relation (i.e. $A \red\cdot_G B$ iff $A \red_G B$ and there exists no $C$ such that $A \red_G C \red_G B$). 

\begin{theorem}
\label{posetTheorem}
The $G$-bounded poset is finite, has a single least element $G$, and its maximal elements are all AVGs. Also, $G' \red\cdot_G G''$ iff there exists a single $G$-bounded reduction operation that transforms $G''$ into $G'$.
\end{theorem}

The proof is given in  Appendix D.

As the $G$-bounded poset is finite, it can be represented by a Hasse diagram whose nodes are the $G$-bounded history graphs and whose edges, which are the covering relation, represent equivalence classes of $G$-bounded operations. 
Figure \ref{hasseDiagram} shows a simple $G$-bounded poset Hasse diagram. 

\begin{figure}[h!]
\begin{center}
\includegraphics[width=10 cm]{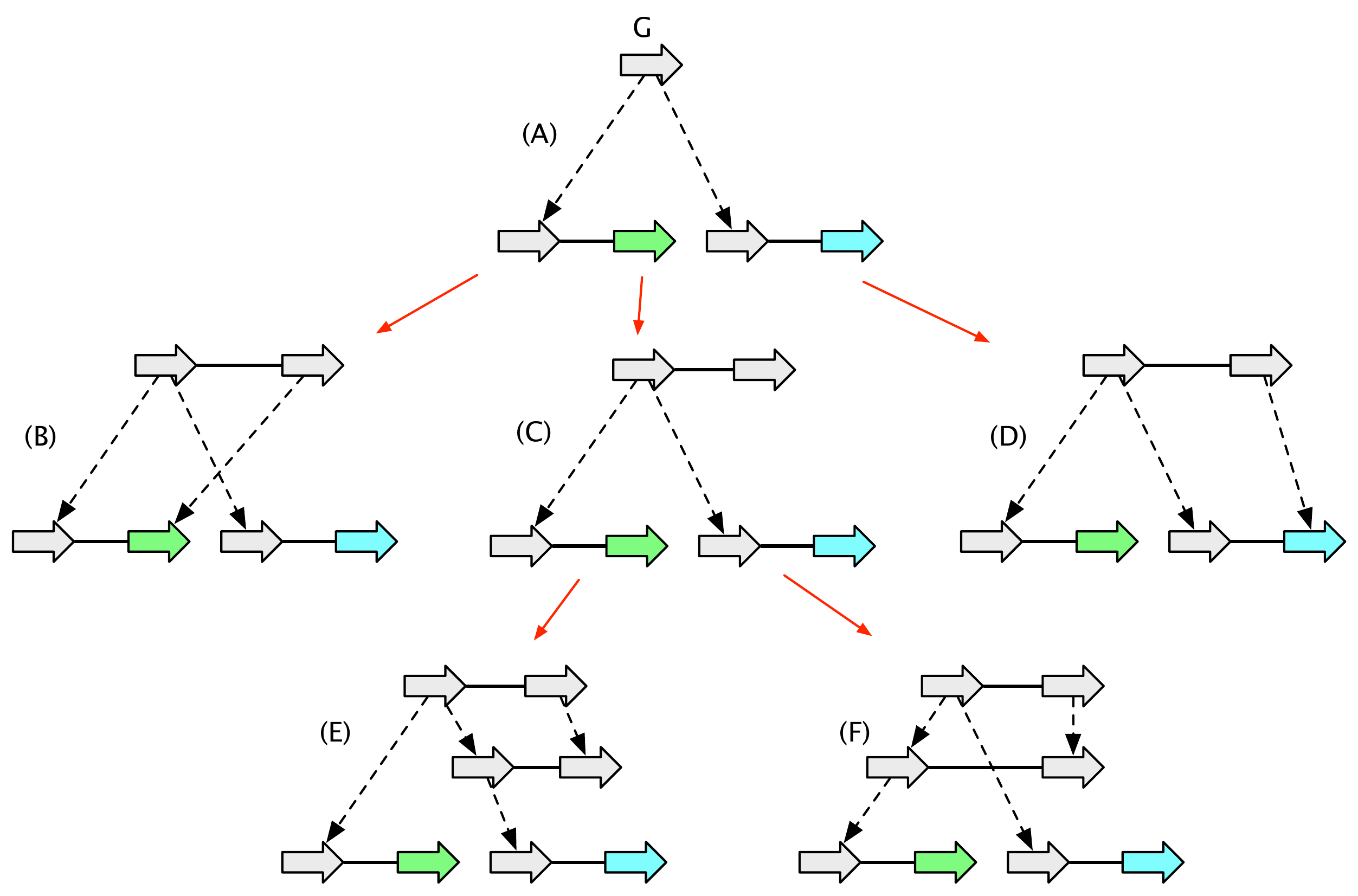}
\caption{A Hasse diagram of the G-bounded poset for an example history graph.}
\label{hasseDiagram}
\end{center}
\end{figure}  

\subsection{A Basic Implementation}

The previous four theorems establish the mechanics of everything we need to sample the finite set of $G$-optimal AVGs, and thus, amongst other things, determine the cost of a history graph. 
Although it will require further work to establish practical and efficient sampling algorithms, we have implemented a simple graph library in Python that for an input history graph $G$ iteratively generates $G$-bounded AVGs (\url{https://github.com/dzerbino/pyAVG}) through sequences of $G$-bounded extension operations. 

To test the library we used simulations. For each simulation we generated a simple history $\textbf{H}$ by forward simulation, starting from a genome with 5 vertices in a single thread and simulating 4 epochs in which either whole chromosome replication or rearrangements occurred and substitutions were made at a constant rate at each branch. The labels in the simulation correspond to single DNA bases. To ensure complexity, we selected histories with substitutions, rearrangements and at least two epochs of replication. 
We created a reduction $G$ of $\textbf{H}$ by removing from $\textbf{H}$ all labels of internal vertices and adjacencies incident on internal vertices and finally contracting the parent branch of all internal vertices. As a result, the reduced history contained only the leaf threads and branch trees that, containing no internal vertices, simply indicate the homologies between the vertices. To simulate incomplete genome assemblies, we randomly removed, on average, 10\% of the adjacencies, labels and vertices from these leaf threads. 
To test our library we enumerated sequences of $G$-bounded history graphs starting at $G$,  at each step picking at random a possible $G$-bounded extension operation until we reached an AVG. 
We sampled 20,000 starts for each of 20 randomly sampled pairs of history and starting graph. To make the search strategy efficient, we restarted the search if we reached an extension with a higher total sum of lower bound substitution and rearrangement costs than $s_u(G) + r_u(G)$, initially, and then subsequently the sum of the substitution and rearrangement costs of the best AVG found up to that point.  Tables  \ref{subsExpTable} and \ref{rearrangeExpTable} show the results of these 20 sampling runs. Figure \ref{samplingExample} shows one example of $\textbf{H}$, $G$ and a sampled AVG.
 
For these simulations the minimum rearrangement cost of any sampled AVG is often close or equal to $r_l(G)$, while the maximum rearrangement cost of any sampled AVG is generally slightly greater than $r_u(G)$. Notably, we found that AVG extensions sometimes had lower cost than the original simple history, this occurring because of the information loss that resulted from reducing $\textbf{H}$ to $G$. 

Repeating these experiments with histories that started with 10 root vertices in the simple history, but which were otherwise simulated identically, demonstrates that the naive random search procedure implemented here fails to find reasonable histories within a set of only 20,000 random samples (data not shown), so, as might be expected, more intelligent sampling strategies will be needed to find parsimonious interpretations of even moderately complex datasets. However, with more efficient sampling algorithms, a history graph sampling algorithm could be applied to find solutions to various established parsimony problems, such as the DCJ median problem, or be used for less explored problems, such as the inference of gene trees incorporating synteny information.


\begin{table}
\centering
\rowcolors{1}{}{lightgray}
\begin{tabular}{cccccccc}
exp. & $s(\textbf{H})$ & $u_s(G)$ & $s_l(G)$ & $s_u(G)$ & $s(H_{smin})$ & $s(H_{smax})$ \\
1 & 3 & 10 & 1 & 1 & 1 & 2\\
2 & 1 & 14 & 1 & 2 & 2 & 3\\
3 & 2 & 15 & 2 & 3 & 3 & 3\\
4 & 3 & 12 & 2 & 2 & 2 & 4\\
5 & 2 & 13 & 2 & 2 & 2 & 4\\
6 & 2 & 12 & 2 & 2 & 2 & 5\\
7 & 2 & 10 & 1 & 1 & 1 & 2\\
8 & 1 & 13 & 1 & 1 & 1 & 2\\
9 & 3 & 11 & 0 & 0 & 0 & 0\\
10 & 4 & 8 & 2 & 2 & 2 & 3\\
11 & 2 & 10 & 2 & 2 & 2 & 3\\
12 & 2 & 13 & 1 & 1 & 1 & 1\\
13 & 2 & 11 & 1 & 2 & 2 & 3\\
14 & 2 & 11 & 2 & 2 & 2 & 4\\
15 & 3 & 14 & 2 & 2 & 2 & 2\\
16 & 2 & 10 & 1 & 1 & 1 & 1\\
17 & 2 & 30 & 1 & 1 & 1 & 1\\
18 & 3 & 13 & 1 & 1 & 1 & 1\\
19 & 2 & 10 & 0 & 0 & 0 & 0\\
20 & 1 & 9 & 1 & 1 & 1 & 1\\
\end{tabular}
\caption{Simulation results assessing substitution ambiguity and cost.  Each row represents a separate initial history.
The cost $s(\textbf{H})$ is the substitution cost of the simple history from which $G$ is derived.  Also given are the ambiguity $u_s(G)$, lower $s_l(G)$, and upper $s_u(G)$ substitution cost bounds for $G$. The minimum and maximum substitution costs of $G$-bounded AVG extensions found by sampling are denoted $s(H_{smin})$ and $s(H_{smax})$, resp.}
\label{subsExpTable}
\end{table}

\begin{table}
\centering
\rowcolors{1}{}{lightgray}
\begin{tabular}{cccccccc}
exp. & $r(\textbf{H})$ & $u_r(G)$ & $r_l(G)$ & $r_u(G)$ & $r(H_{rmin})$ & $r(H_{rmax})$ \\
1 & 2 & 12 & 2 & 10 & 2 & 9 \\
2 & 2 & 20 & 2 & 14 & 2 & 14 \\
3 & 2 & 20 & 2 & 14 & 2 & 12 \\
4 & 2 & 20 & 2 & 14 & 2 & 14 \\
5 & 2 & 18 & 1 & 13 & 1 & 11 \\
6 & 2 & 8 & 2 & 7 & 2 & 6 \\
7 & 2 & 8 & 0 & 7 & 0 & 4 \\
8 & 2 & 18 & 1 & 13 & 2 & 10 \\
9 & 2 & 10 & 1 & 7 & 1 & 7 \\
10 & 2 & 14 & 0 & 11 & 0 & 8 \\
11 & 2 & 6 & 0 & 6 & 0 & 4  \\
12 & 2 & 6 & 1 & 7 & 1 & 4  \\
13 & 2 & 16 & 0 & 12 & 0 & 9  \\
14 & 2 & 20 & 2 & 14 & 4 & 12  \\
15 & 2 & 20 & 1 & 14 & 1 & 10  \\
16 & 2 & 6 & 0 & 5 & 0 & 5  \\
17 & 1 & 26 & 1 & 17 & 1 & 13  \\
18 & 2 & 18 & 1 & 13 & 1 & 11  \\
19 & 2 & 6 & 0 & 6 & 0 & 5  \\
20 & 2 & 4 & 2 & 5 & 2 & 2  \\
\end{tabular}
\caption{Simulation results assessing rearrangement ambiguity and cost.  Each row represents a separate initial history.
The cost $r(\textbf{H})$ is the rearrangement cost of the simple history from which $G$ is derived.  Also given are the ambiguity $u_r(G)$, lower $r_l(G)$, and upper $r_u(G)$ rearrangement cost bounds for $G$. The minimum and maximum rearrangement costs of $G$-bounded AVG extensions found by sampling are denoted $r(H_{rmin})$ and $r(H_{rmax})$, resp.}
\label{rearrangeExpTable}
\end{table}

\begin{figure}[h!]
\begin{center}
\includegraphics[width=13cm]{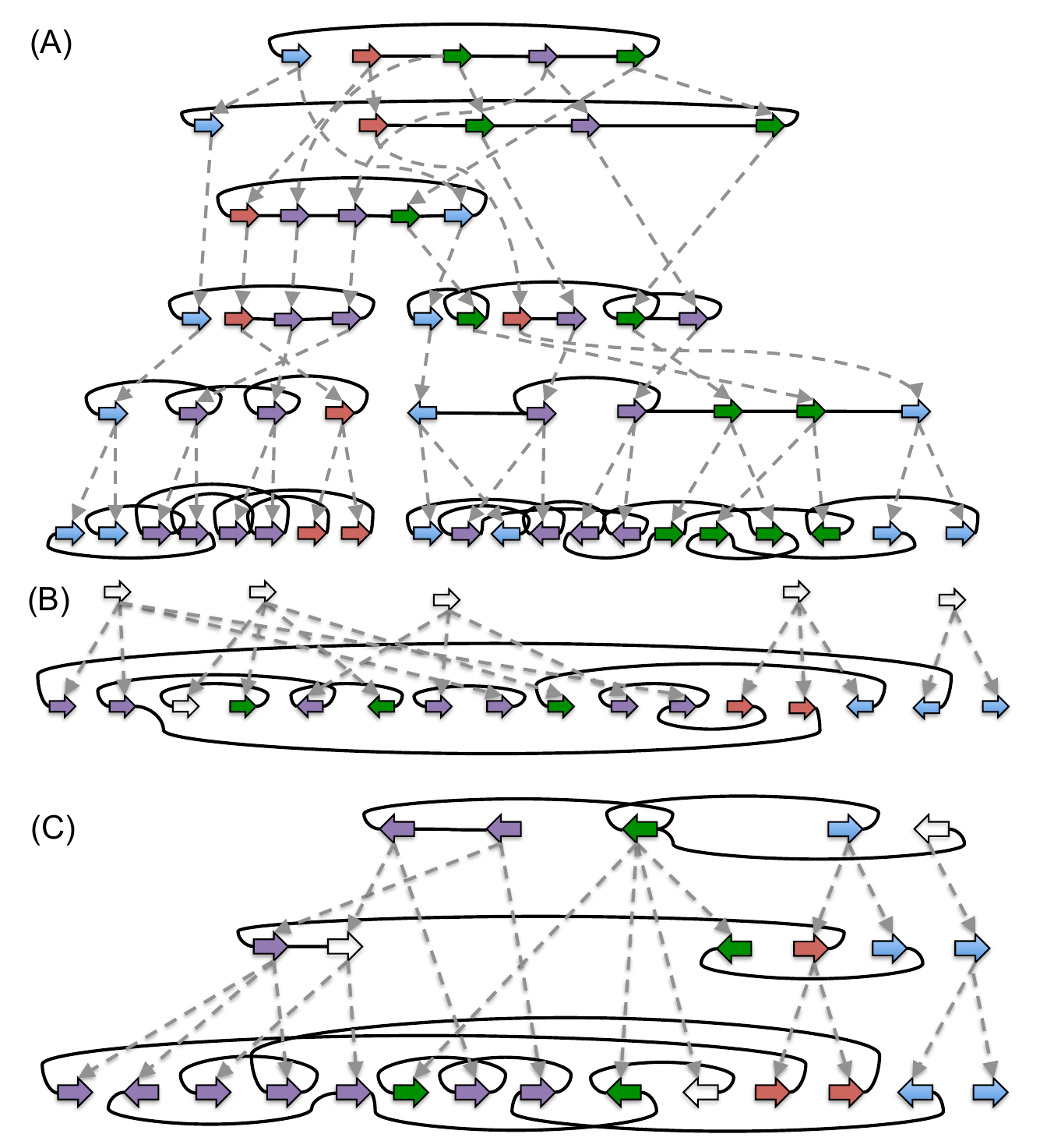}
\caption{History graph examples generated by simulation. \textbf{(A)} \textbf{H}, \textbf{(B)} $G$, \textbf{(C)} An example of $H_{rmin}$ and  $H_{smin}$. Example corresponds to experiment 1 in Tables \ref{subsExpTable} and \ref{rearrangeExpTable}. The $G$-bounded extension sequence from $G$ to this AVG involved the creation of just 7 adjacencies, 5 vertices and 7 labels. Graph layouts were computationally derived.}
\label{samplingExample}
\end{center}
\end{figure}

\section{Discussion}

We have introduced a general model for genome evolution under parsimony, but the reduction relation and the definition of the $G$-bounded set may appear arbitrary. We highlight below the reasons for our choice of reduction relation, how reduction relates to other orderings over graphs, and how we can easily approximate a set of $G$-reducible elements, something critical to the sampling of $G$-bounded extensions of a given graph. We then briefly discuss the possibilities of yet more compact graphical representations.

In the reduction relation, we allow the deletion of vertices, vertex labels and adjacencies, but forbid branch deletion. Otherwise, extensions would allow the invention of homology between vertices (see Figure \ref{justifyingReduction}(A)).
Conversely, branches can be contracted but not adjacencies, otherwise extensions could create interstitial vertices without any rearrangement (see Figure \ref{justifyingReduction}(B)). 

We disallow the non-trivial contraction of the incoming branch of attached or labeled vertices, with the one exception for branches with free-parents, because it would allow a reduction to merge previously separate threads (see Figure \ref{justifyingReduction}(C)), and because vertices could be reduced to become ancestors of originally indirectly related vertices (see Figure \ref{justifyingReduction}(D)). We allow the one exception for the contraction of the incoming branch of attached or labeled vertices when the branch has a free-parent because disallowing it would forbid reductions that removed information from root vertices (see Figure \ref{justifyingReduction}(E)) and allowing it does not permit the issues highlighted in Figures \ref{justifyingReduction}(C-D).

It is informative to consider the relationship between reduction operations and the reduction relation.
When a graph contains multiple copies of isomorphic structures, distinct reduction operations can result in isomorphic reductions (see Figure \ref{justifyingReduction}(F-I)), therefore each possible reduction in the covering set (transitive reduction) of the reduction relation represents an equivalence class of reduction operations. 

A \emph{valid permutation} of a reduction sequence is a permutation in which all operations remain reduction operations when performed in sequence. 
Clearly not all permutations of a reduction sequence have this property, however the following lemma illustrates the relationship between valid permutations.

\begin{lemma} \label{validPermutationsLemma}
All valid permutations of a reduction sequence create isomorphic reductions.
\end{lemma}


Reduction is somewhat analogous to a restricted form of the graph minor. 
Importantly, the graph minor is a well-quasi-ordering (WQO) (\cite{Bienstock:1994tw}), i.e. in any infinite set of graphs there exists a pair such that one is the minor of the other. 

\begin{lemma}
Reduction is not a WQO. 
\end{lemma}

\begin{proof}
Consider the infinite set of cyclic threads, they are not reductions of one another. 
\end{proof}

An ordering is a WQO if every set has a finite subset of minimal elements. In contrast, it can be shown that for the reduction relation, even the set of AVG extensions of a single base history $G$ can have an infinite set of minimal elements.


\begin{lemma} \label{infiniteMinimalExtensions}
There exists a history graph $G$ with an infinite number of $G$-minimal extensions. 
\end{lemma}

The proof is given in  Appendix E.


\begin{figure}[h!]
\begin{center}
\includegraphics[width=13cm]{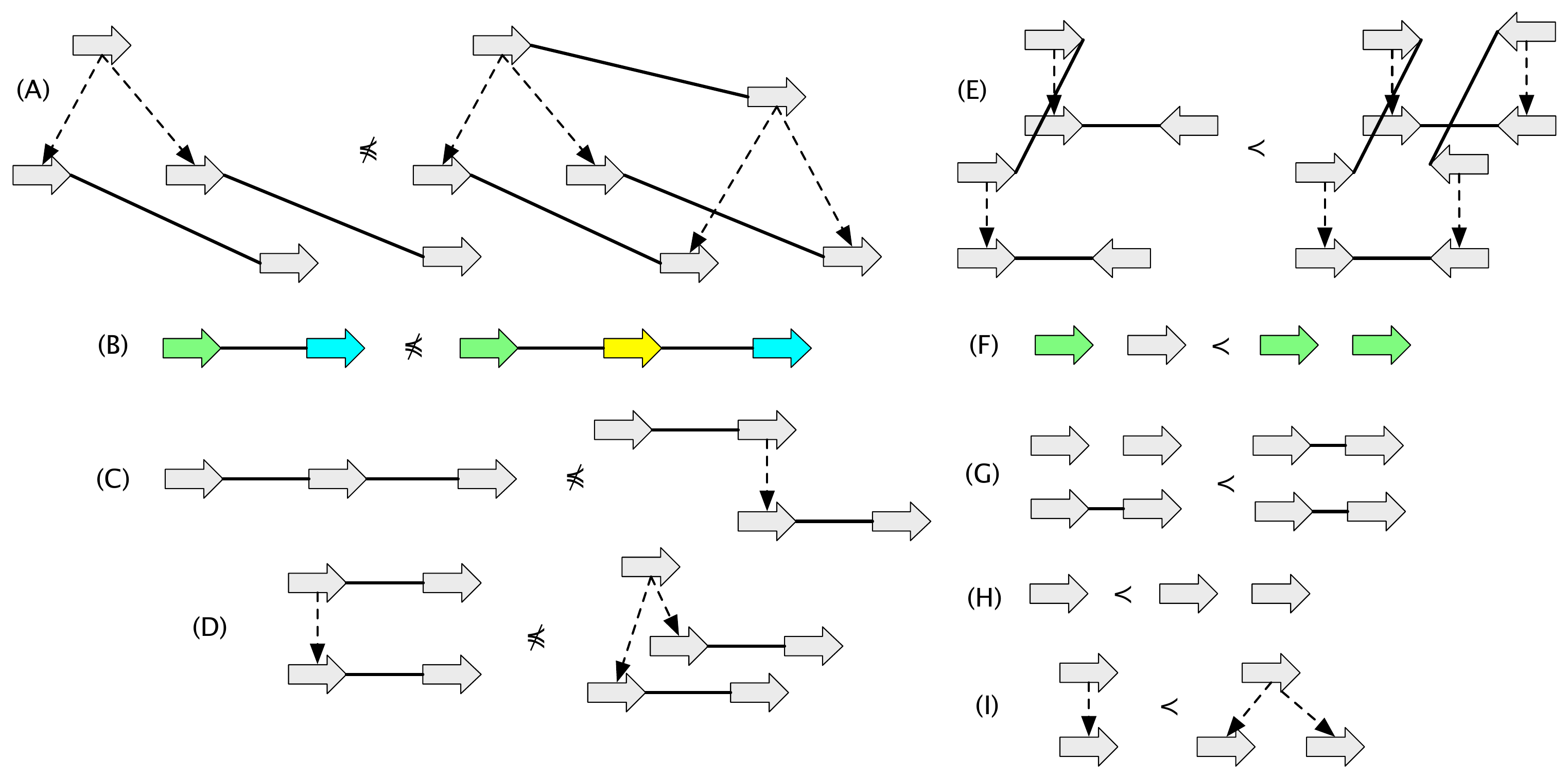}
\caption{\textbf{(A,B,C,D)} The graphs on the left side are not reductions of the graphs on the right. \textbf{(E)} The graph on the left is a reduction of the graph on the right. \textbf{(F,G,H,I)} Examples of equivalence classes of reduction operations, where multiple distinct reduction operations result in the same reduction.}
\label{justifyingReduction}
\end{center}
\end{figure}

One barrier to exploring the $G$-bounded poset is deciding for a pair of history graphs $G$ and $G'$ such that $G \redeq G'$ if an element is $G$-reducible. This problem is of unknown complexity, and may well be NP-hard.
To avoid the potential complexity of this problem we can define an alternative notion of reducibility.
A \emph{fix} for $(G, G')$, where $G \redeq G'$, is a history subgraph of $(V_{G'}, E_{G'}, B_{G'}^+)$  isomorphic to $G$, where $B_{G'}^+$ is the transitive closure of $B_{G'}$.
Starting from an input history graph $G$ and a fix isomorphic to it, we can easily update the fix as we create extensions of $G$. 
For an extension of $G$, elements in the fix become the equivalent of $G$-irreducible, while elements not in the fix become the equivalent of $G$-reducible. From a starting graph we can therefore explore a completely analogous version of $G$-bounded, replacing the question of $G$-reducibilty with membership of the fix.


Following from Lemma \ref{validPermutationsLemma}, there is a bijection between the set of fixes for $G \redeq G'$ and the set of equivalence classes of reduction sequences that are all valid permutations of each other. This is the limitation of considering membership of a fix instead of assessing if an element is $G$-reducible, it limits us to considering only a single equivalence class of reduction sequences in exploring the analogous poset to $G$-bounded.

It is in general possible to reduce the size of the set $G$-bounded while still maintaining the properties that it can be efficiently sampled and contains $G$-optimal. However, this is likely to be at the expense of making the definition of $G$-bounded more complex. One approach is to add further ``forbidden configurations" to the definition of $G$-bounded, like the $G$-reducible ping adjacencies that are forbidden in the current definition of $G$-bounded. Forbidding these was essential to making $G$-bounded finite, but we might consider also forbidding other configurations just to make $G$-bounded smaller.

It is possible to consider a graph representation of histories that use fewer vertex nodes if we are willing to allow for the possibility that a subrange of the sequence of a vertex be ancestral to a subrange of the sequence of another vertex. This is a common approach in ancestral recombination graphs (\cite{Song:2005vj}). Such a representation entails the additional complexity of needing to specify the sequence subranges for every branch, but may in some applications be a worthwhile trade off for reducing the number of vertices in the graph. The theory of such graphs is mathematically equivalent to the theory of the history graphs presented here, but the implementation would differ.


\section{Conclusion}

We have introduced a graph model in which a set of chromosomes evolves via the processes of whole chromosome replication, gain and loss, substitution and DCJ rearrangements. We have demonstrated upper and lower bounds on maximum parsimony cost that are trivial to compute despite the intractability of the underlying problem. Though these cost bounding functions are relatively crude and can almost certainly be tightened for many cases, they become tight for AVGs. This implies that we only need to reach AVG extensions to assess cost when sampling extensions.

To our knowledge, this is the first fully general model of chromosome evolution by substitution, replication, and rearrangement. However, it has its limitations. For example, it treats common rearrangements, such as recombinations and indels as any other rearrangement, and only takes into account maximum parsimony evolutionary histories. We anticipate future extensions that incorporate more nuanced cost functions, as well as probabilistic models over all possible histories.

The constructive definition of the $G$-bounded poset, coupled with the upper and lower bound functions, suggests simple branch and bound based sampling algorithms for exploring low-cost genome histories. To facilitate the practical exploration of the space of optimal and near optimal genome histories, we expect that more advanced sampling strategies across the $G$-bounded poset could be devised.

\section{Methods}
 
\subsection{Appendix A}

In this section we prove Theorem \ref{avgTheorem}.


We first define some convenient notations to describe lifted labels and edges.
For a vertex $x$ let $L'_x = (L_x, N_x)$ be its multiset of lifted labels, where $L(x)$ is the set of distinct lifted labels for $x$, and for each lifted label $\rho$, $N_x(\rho)$ is the number of times $\rho$ appears as a lifted label for $x$, i.e. $L_x = \{ l(y) : A(y) = x \}\subseteq \Sigma^*$ and $N_x : L_x \rightarrow \mathbb{Z}_+$ such that $N_x(\rho) = | \{ y : A(y) = x , l(y) = \rho \} |$. 

For a side $x_{\alpha}$, and overloading notation, let $L'_{x_{\alpha}} = (L_{x_{\alpha}}, N_{x_{\alpha}})$ be its multiset of lifted edges, where $L(x_\alpha)$ is the set of distinct lifted adjacencies incident with $x_\alpha$, and for each lifted adjacency $\{ x_\alpha, w_\gamma \}$, $N_{x_\alpha}(\{ x_\alpha, w_\gamma \})$ is the number of sides whose lifting ancestor is $x_\alpha$, and which are connected by an adjacency to a side whose lifting ancestor is $w_\gamma$, i.e. $L_{x_{\alpha}} = \{ \{ x_\alpha=A(y_\alpha), A(z_\beta) \} : \{ y_\alpha, z_\beta \} \in E_G \}$ and $N_{x_\alpha} = L_{x_{\alpha}} \rightarrow \mathbb{Z}_+$ such that $N_{x_\alpha}(\{ x_\alpha, w_\gamma \}) = | \{ y_\alpha : \{ x_\alpha = A(y_\alpha), w_\gamma \} \in L_{x_\alpha} \} |$. 

Note that for a side $x_\alpha$, $N_{x_{\alpha}}(\{x_\alpha, w_\gamma \})$ gives the multiplicity of lifted adjacency incidences with $x_\alpha$, not the multiplicity of $\{x_\alpha, w_\gamma \}$. In particular, if two sides $x_\alpha$ and $x'_\alpha$ are attached and share the same lifting ancestor $A(x_\alpha)$, then $N_{A(x_{\alpha})}(\{A(x_\alpha), A(x_\alpha) \})$ is incremented by 2. On the contrary, if $x_\alpha$ is connected to $w_\gamma$ and $A(x_\alpha)$ is distinct from $A(w_\gamma)$, then both $N_{A(x_{\alpha})}(\{A(x_\alpha), A(w_\gamma) \})$ and $N_{A(w_{\gamma)}}(\{A(x_\alpha), A(w_\gamma) \})$ are incremented by 1.

For a vertex (resp. side) $x$ the multi-set of non-trivial lifted labels (adjacencies) is $\tilde{L}'_{x} = (\tilde{L}_x, \tilde{N}_x) \subseteq L'_{x}$. 

\subsubsection*{The Equivalence of LBSC to UBSC and LBRC to UBRC for AVGs}

\begin{lemma} \label{equivalentSubstitutionBoundsLemma}
For any AVG $H$, $s_l(H) = s_u(H)$.
\end{lemma}

\begin{proof}
For a vertex $x$ without substitution ambiguity there is at most one non-trivial lifted label, that, if it exists, has a multiplicity of one, therefore $|\tilde{L}_x| = |\tilde{L}'_{x}| = 0$ or $1$. Let $\delta_{a,b}$ be the Kronecker delta, i.e. $\delta_{a,b} = 1$ if $a = b$, else $0$. It is easily verified for every possible case:
\[ max(0, |\tilde{L}_x| - \delta_{l(x),\emptyset}) =   |\tilde{L}'_{x}| - \mbox{ }\delta_{l(x),\emptyset} \times \mbox{ } \underset{\rho \in \tilde{L}'_{x}}{\max\mbox{ }} N_{x}(\rho), \] summing over modules, therefore: \[ s_l(H) = \sum_{x \in V_{\textbf{L}(H)}} max(0, |\tilde{L}_x| - \delta_{l(x),\emptyset}) = \sum_{x \in V_{\textbf{L}(H)}} |\tilde{L}'_{x}| - \mbox{ }\delta_{l(x),\emptyset} \times \mbox{ } \underset{\rho \in \tilde{L}'_{x}}{\max\mbox{ }} N_{x}(\rho) = s_u(H). \]
\end{proof}

A module is \emph{simple} if each side has at most one incidence with a non-trivial lifted adjacency.

\begin{lemma}
All modules in an AVG are simple.
\end{lemma}

\begin{proof}
Follows from definition of rearrangement ambiguity.
\end{proof}



\begin{lemma} \label{equivalentRearrangementBoundsLemma}
For an AVG $H$, $r_l(H) = r_u(H)$.
\end{lemma}

\begin{proof}
Let $M$ be a simple module and let $k_M =  \sum_{x_{\alpha} \in V_M}  \delta_{1,|\tilde{L}'_{x_{\alpha}}|}$, i.e. the number of sides in $V_M$ with a single incidence with a non-trivial lift. 

As the module is simple it is a path or a cycle, and hence $|V_M| - k_M = 0,  1$ or $2$, from which it is easily verified that: \[ \lceil \frac{|V_M| - k_M}{2} \rceil - 1 = -\prod_{x_{\alpha} \in V_M} \delta_{1,|\tilde{L}'_{x_{\alpha}} |}. \]

Summing over modules in H, which are all simple, therefore:  \[ \sum_{M \in M(H)} \lceil \frac{|V_M| - k_M}{2} \rceil - 1 = \sum_{M \in M(H)} -\prod_{x_{\alpha} \in V_M} \delta_{1,|\tilde{L}'_{x_{\alpha}} |}. \]

As all modules of $H$ are simple, $k_M$ is always even and $k_M = \sum_{x_{\alpha} \in V_M} | \tilde{L}'_{x_{\alpha}}|$, therefore: 

\[\sum_{M \in M(H)} (\lceil \frac{|V_M|}{2} \rceil - 1) - \frac{1}{2}\sum_{x_{\alpha} \in V_M} | \tilde{L}'_{x_{\alpha}}| = \sum_{M \in M(H)} -\prod_{x_{\alpha} \in V_M} \delta_{1,|\tilde{L}'_{x_{\alpha}} |}, \] therefore, for an AVG $H$

\[r_l(H) = \sum_{M \in M(H)} (\lceil \frac{|V_M|}{2} \rceil - 1) = \sum_{M \in M(H)} \left( \frac{1}{2}\sum_{x_{\alpha} \in V_M} | \tilde{L}'_{x_{\alpha}}| -\prod_{x_{\alpha} \in V_M} \delta_{1,|\tilde{L}'_{x_{\alpha}} |} \right)  = r_u(H). \]

\end{proof}

\subsubsection*{A Bounded Transformation of a History Graph into an AVG}

In this section we will prove that any history graph $G$ has an AVG extension $H$ such that $s_u(G) \ge s_u(H)$ and $r_u(G) \ge r_u(H)$. To do this we define sequences of extension operations that when applied iteratively and exhaustively construct such an extension.

A vertex or side $x$ is \emph{ambiguous} if $|\tilde{L}'_{x} |> 1$.
For an ambiguous free-root $x'$ and unlabeled root vertex $x$ such that $A(x) = x'$, a \emph{root labeling extension} is a labeling of $x$ with a member of the set $\underset{\rho \in \tilde{L}_{x'}}{\arg\max\mbox{ }} N_{x'}(\rho)$ (See Figure  \ref{appendixDNAToAVG}(A)).

\begin{lemma}
For any history graph $G$ containing an ambiguous free-root there exists a root labeling extension $G'$ of $G$ such that $s_u(G) = s_u(G')$, $r_u(G) = r_u(G')$ and $u(G) > u(G')$. 
\end{lemma}


For a branch $(x, x')$ an \emph{interpolation} is the extension resulting from the creation of a new vertex $x''$ and branches $(x, x'')$ and $(x'', x')$ and the deletion of $(x, x')$. 
Let $x$ be a labeled and ambiguous vertex and $x'$ be a labeled vertex such that $A(x') = x$ and $l(x) \not= l(x')$. A \emph{substitution ambiguity reducing extension} is the interpolation of a vertex $x''$ along the parent branch of $x'$ labeled with $l(x)$ (See Figure \ref{appendixDNAToAVG}(B)).

\begin{lemma}
For any history graph $G$ containing no ambiguous free-roots and such that $u_s(G) > 0$, there exists a substitution ambiguity reducing extension $G'$ of $G$ such that $s_u(G) = s_u(G')$, $r_u(G) = r_u(G')$ and $u(G) > u(G')$.
\end{lemma}


The following is used for eliminating rearrangement ambiguity.
For an unattached junction side $x_{\alpha}$ a \emph{junction side attachment extension} is the extension resulting from the following: If $x_{\alpha}$ has no attached ancestor, the creation of a new vertex and adjacency connecting a side of the new vertex to $x_{\alpha}$ (see Figure \ref{appendixDNAToAVG}(C) for an example), else $\{ A(x_{\alpha}), y_{\beta} \} \in E_G$ and the extension is the creation of a new vertex $y'$, branch $(y, y')$ and adjacency $\{ x_{\alpha}, y'_{\beta} \}$ (See Figure \ref{appendixDNAToAVG}(D)).

\begin{lemma}
For any history graph $G$ containing an unattached junction side, there exists a junction side attachment extension $G'$ of $G$ such that $s_u(G) = s_u(G')$, $r_u(G) \ge r_u(G')$, $u(G) \ge u(G')$ and $G'$ contains one less unattached junction side than $G$.
\end{lemma}


Let $\{ x_{\alpha}, y_\beta \} $ and $\{ A(x_\alpha), z_\gamma \}$ be a pair of adjacencies and $A(x_\alpha)$ be ambiguous. A \emph{rearrangement ambiguity reducing extension} is the interpolation along the parent branch of $x$ a vertex $x'$, the creation of a new vertex $z'$, new branch $(z, z')$ and new adjacency $\{ x'_{\alpha}, z'_{\gamma} \}$ (See Figure \ref{appendixDNAToAVG}(E)).

\begin{lemma}
For any history graph $G$ containing no unattached junction sides and such that $u_r(G) > 0$, there exists a rearrangement ambiguity reducing extension $G'$ of $G$ such that $s_u(G) = s_u(G')$, $r_u(G) \ge r_u(G')$ and $u(G) > u(G')$.
\end{lemma}


\begin{figure}[h!]
\begin{center}
\includegraphics[width=10cm]{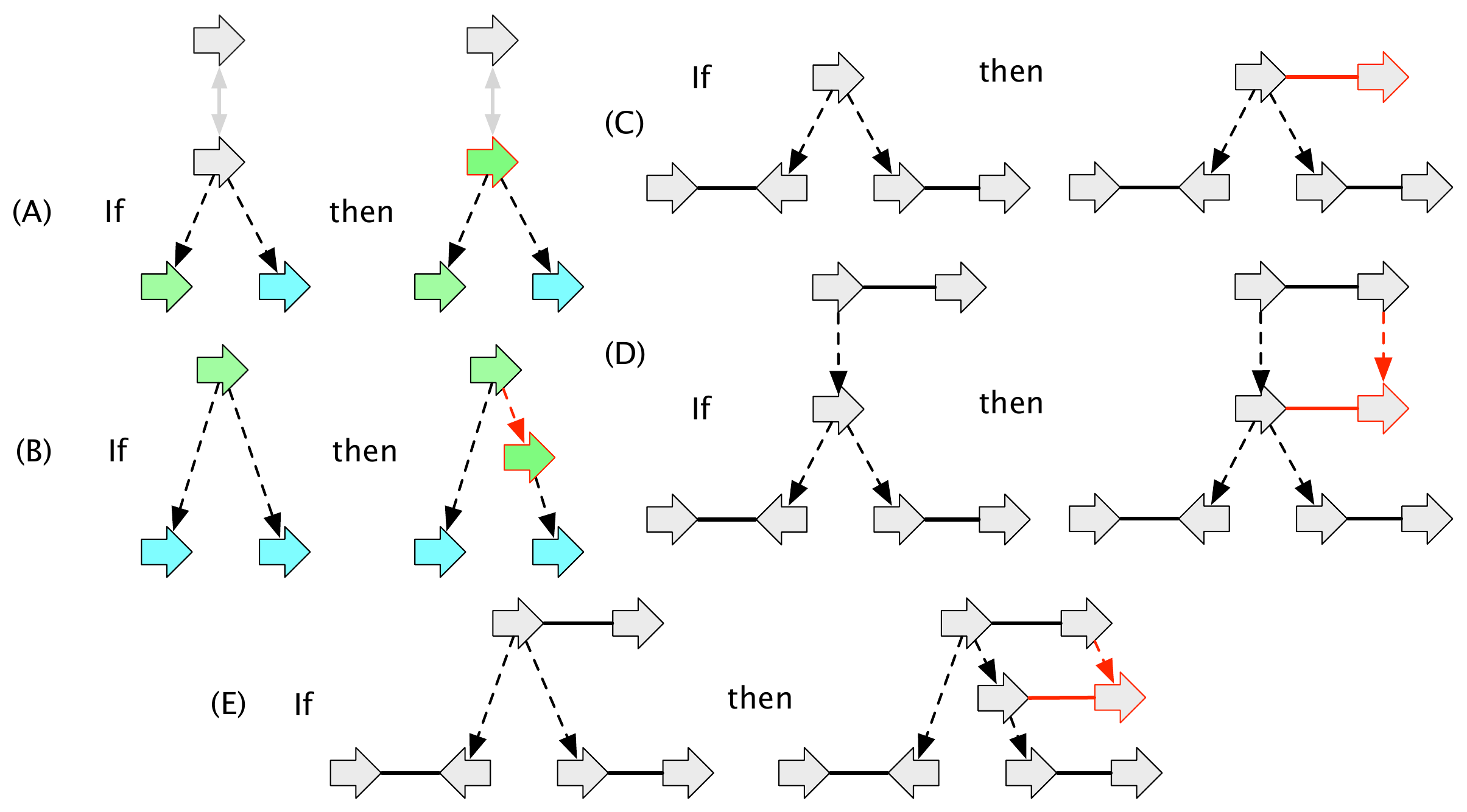}
\caption{(A) A root labeling extension. (B) A substitution ambiguity reducing extension. (C-D) Junction side attachment extensions. (E) A rearrangement ambiguity reducing extensions. Elements in red / outlined in red are those added in extension.}
\label{appendixDNAToAVG}
\end{center}
\end{figure}

We can now prove the desired lemma.

\begin{lemma} \label{dnaToAVGLemma}
Any history graph $G$ has an AVG extension $H$ such that $s_u(G) \ge s_u(H)$ and $r_u(G) \ge r_u(H)$.
\end{lemma}

\begin{proof}

Using the previous 4 lemmas it is easily verified the result of the following algorithm is an AVG extension $H$ for a history graph $G$ such that $s_u(G) \ge s_u(H)$ and $r_u(G) \ge r_u(H)$.

\begin{algorithmic}
\STATE $H \leftarrow G$
\WHILE {$u(H) > 0$}
\IF {$H$ contains an ambiguous free-root}
\STATE $H \leftarrow$ root labeling extension of $H$.
\ELSE
\IF {$u_s(H) > 0$}
\STATE $H \leftarrow$ substitution ambiguity reducing extension of $H$.
\ELSE
\IF {$H$ contains an unattached junction side}
\STATE $H \leftarrow$ junction side attachment extension of $H$.
\ELSE
\STATE $H \leftarrow$ rearrangement ambiguity reducing extension of $H$.
\ENDIF
\ENDIF
\ENDIF
\ENDWHILE
\end{algorithmic}
\end{proof}

\subsubsection*{A Bounded Transformation of an AVG into a Realisation}

In this section we will prove that any AVG $H$ has a realisation $\textbf{H}$ such that $s_l(H) = s(\textbf{H})$ and $r_l(H) = r(\textbf{H})$.


A vertex connected by an adjacency to another vertex with more child branches has  \emph{missing children}.
A root vertex that is connected to a non-root vertex has a \emph{missing parent}. 
Missing parents and missing children are collectively \emph{missing branches}.
An unattached side with homologous attached sides has a \emph{missing adjacency}.

We will define a series of extension types that when combined iteratively create an extension in which all vertices are labeled and no elements have missing adjacencies or branches.  For each extension type defined below Figure \ref{appendixAVGToRealisation} shows an example.

For an attached root vertex $x$, the creation of a new vertex $x'$ and branch $(x', x)$ is a \emph{case 1 extension}.
The case 1 extension is used iteratively to initially ensure all roots are unattached. 

For an attached leaf vertex, the creation of a new vertex $x'$ and branch $(x, x')$ is a \emph{case 2 extension}.
The case 2 extension is used iteratively to initially ensure all leaves are unattached. 

For a side $x_{\alpha}$ if $A(x_{\alpha})$ is in a module $M$, $x_{\alpha}$ is in the \emph{face} of $M$. 
Let $M$ be a simple module containing an odd number of sides and let $x_{\alpha}$ be an unattached root side in the face of $M$. The following is a \emph{case 3 extension}: the creation of a pair of vertices $y$ and $y'$, an adjacency connecting a side of $y$ to $x_{\alpha}$ and the branch $(y, y')$. 
The case 3 extension is used iteratively to ensure all modules contain an even number of sides. 

Similarly to vertices and sides, a thread $X$ is \emph{ancestral} to a thread $Y$ in a history graph $G$, and reversely $Y$ is a \emph{descendant} of $X$, if there exists a directed path in $D(G)$ from the vertex representing $X$ to the vertex representing $Y$, otherwise two threads are \emph{unrelated} if they do not have an ancestor/descendant relationship.
For a vertex $x$, $T(x)$ is the thread it is part of.
For a pair of unattached root sides $x_{\alpha}$ and $y_{\beta}$ in the face of a simple module such that $T(x) = T(y)$ or $T(x)$ and $T(y)$ are unrelated, the creation of a new adjacency $\{ x_{\alpha}, y_{\beta} \}$ is a \emph{case 4 extension}.
The case 4 extension is used iteratively to ensure all modules contain attached root sides. 

Let $x_{\alpha}$ be a side in the face of a simple module $M$ such that $x_{\alpha}$ is internal, unattached and has an attached parent. Let $(y, y')$ be a branch such that $y'_{\beta}$ is a side in the face of $M$, $T(y)$ is not descendant of $T(x)$, if $T(y) = T(x)$ then $y$ is unattached, $T(y')$ is descendant or unrelated to $T(x)$, and the sides $A(x_{\alpha})$ and $A(y'_{\beta})$ in $M$ are connected by a path containing an odd number of adjacencies/lifted adjacencies. If $y_{\beta}$ is unattached and $T(y)$ is unrelated or equal to $T(x)$ then the creation of the adjacency $\{ x_{\alpha}, y_{\beta} \}$ is the \emph{case 5 extension}, else the interpolation of a vertex $y''$ on the branch $(y, y')$ and creation of the adjacency $\{ x_{\alpha}, y_{\beta}'' \}$ is the \emph{case 5 extension}.
The case 5 extension is used iteratively to ensure all internal vertices are attached.

For an adjacency $\{ x_{\alpha}, y_{\beta} \}$ such that $y$ has fewer children than $x$, the creation of a new vertex $y'$ and branch $(y, y')$ is a \emph{case 6 extension}. 
The case 6 extension is used iteratively to ensure there are no vertices with missing children. 

Let $x_{\alpha}$ and $y_{\beta}$ be a pair of unattached leaf sides in the face of a simple module $M$ such that $T(x)$ and $T(y)$ are unrelated or equal, $A(x_{\alpha})$ and $A(y_\beta)$ are attached and are either connected by an adjacency or both not incident with a non-trivial lifted adjacency. The creation of a new adjacency $\{ x_{\alpha}, y_{\beta} \}$ is a \emph{case 7 extension}.
The case 7 extension is used iteratively to ensure there are no leaf vertices with missing adjacencies.

For a branch-tree containing no labeled vertices, the labeling of any single vertex in the branch-tree with a member of $\Sigma^*$ is a \emph{case 8 extension}.
For a branch $(x, y)$, such that $y$ is labeled and $x$ is unlabeled the labeling of $x$ with the label of $y$ is a \emph{case 9 extension}.
For a branch $(x, y)$, such that $x$ is labeled and $y$ is unlabeled the labeling of $y$ with the label of $x$ is a \emph{case 10 extension}.
The case 8, 9 and 10 extensions are used iteratively to ensure there are no unlabeled vertices. 

\begin{figure}[h!]
\begin{center}
\includegraphics[width=10cm]{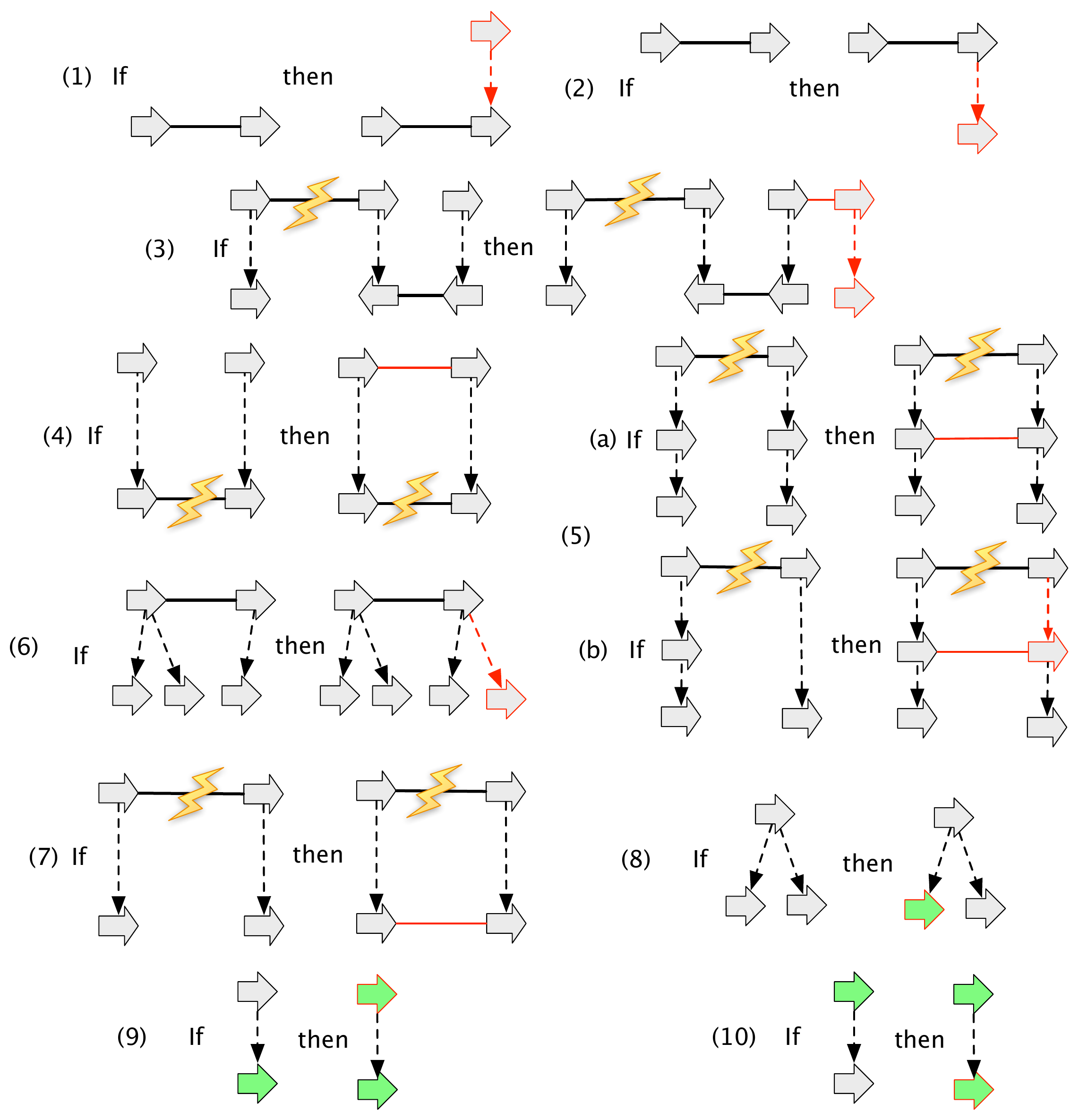}
\caption{Case 1 to 10 extensions. Adjacencies with lightning bolts may be expanded to include additional elements within the simple module. Elements in red / outlined in red are those added in extension.}
\label{appendixAVGToRealisation}
\end{center}
\end{figure}

\begin{lemma}
\label{extensionIsAvg}

For an AVG $H$, if $H'$ is obtained from $H$ by any of the 10 extensions cases above then  $s_l(H) = s_l(H')$ and $r_l(H) = r_l(H')$.
\end{lemma}


\begin{lemma}
\label{finiteExtensionsLemma}
For an AVG $H$, each of the ten types of extensions above can only be applied consecutively a finite number of times until there are no more opportunities in the graph to apply an extension of that type. 
\end{lemma}




\begin{lemma} \label{noMissingLemma}
Any AVG $H$ has an AVG extension $H'$ with no missing labels, adjacencies or branches and such that $s_l(H) = s_l(H')$ and $r_l(H) = r_l(H')$.
\end{lemma}

\begin{proof}
We will demonstrate that the following algorithm converts an AVG into an AVG with no missing adjacencies or branches or unlabeled vertices.

\begin{algorithmic}   
\STATE $H' \leftarrow H$
\STATE $i \leftarrow 1$
\WHILE{$i \le 10$}
\WHILE{$H' \mbox{ has a case } i \mbox{ extension}$}
\STATE $ H' \leftarrow$ case $i$ extension of $H'$
\ENDWHILE
\STATE $i \leftarrow i + 1$
\ENDWHILE
\end{algorithmic}  

It follows from Lemma \ref{finiteExtensionsLemma} that the algorithm always terminates and from Lemma \ref{extensionIsAvg} that $H'$ is an AVG such that $s_l(H) = s_l(H')$ and $r_l(H) = r_l(H')$. 

It remains to prove that $H'$ has no missing branches or adjacencies or unlabeled vertices. Call the AVG extension resulting at the end of the ith loop of line 5 of the algorithm the \emph{case i complete extension}. The following series of compounding statements are straightforward to verify. 
\begin{itemize}
\item The case 3 complete extension contains no modules with an odd number of sides. 

The case 2 extensions ensure that all root vertices are unattached, and every case 3 extension attaches a root vertex in a module with an odd number of sides to a newly created root vertex, so ensuring the module contains an even number of sides, so for every module with an odd number of sides there exists a case 3 extension. 
\item The case 4 complete extension additionally contains no root sides with missing adjacencies or root vertices with missing parents. 

The case 3 extensions ensure that there always 0 or 2 unattached root sides in a module, so any unattached root side in a module always has a potential unattached partner root side within the module. The requirement that sides connected in a case 4 extension be in the same or unrelated threads prior to connection does not prevent any root side within the face of a module from becoming attached, because the case 1 extensions ensure that all root vertices are unrelated, the case 2 extensions do not effect root vertices and the case 3 and 4 vertices only result in root vertices being connected to one another.
\item The case 5 complete extension additionally contains no internal vertices with missing adjacencies. 

The case 4 extensions ensure that all root sides within modules are attached. The case 2 extensions ensure that all attached sides have children and the case 3, 4 and 5 extensions ensure this remains true. Given this, and that every module has an even number of sides within it (as a case 3 complete extension), it is straightforward to verify that there is always a case 5 extension in a sequence of such extensions for any internal side within the face of a module.
\item The case 6 complete extension additionally contains no vertices with missing child branches. 
\item The case 7 complete extension additionally contains no leaf sides with missing adjacencies, and therefore has no missing branches or adjacencies. 

Analogously with the case 4 extensions, the requirement that sides connected in a case 7 extension be in the same or unrelated threads does not prevent any leaf side within the face of a module from becoming attached by a case 7 extension, this is because the case 2 extensions ensure all leaf vertices are unrelated, the case 3, 4, 5 and 6 extensions do not connect leaf vertices, and the case 7 extensions only connect leaf sides to one another.
\item The case 8 complete extension additionally contains no branch-trees without any labeled vertices. 
\item The case 9 complete extension additionally contains no unlabeled ancestral vertices that have labeled descendants. 
\item The case 10 complete extension additionally contains no unlabeled vertices, and therefore has no missing adjacencies, branches or labels. 
\end{itemize}
\end{proof}

We can now prove the desired lemma.

\begin{lemma} \label{avgToRealisation}
Any AVG $H$ has a realisation $\textbf{H}$ such that $s_l(H) = s(\textbf{H})$ and $r_l(H) = r(\textbf{H})$.
\end{lemma}

\begin{proof}
Lemma \ref{noMissingLemma} demonstrates there exists an AVG extension $H'$ of $H$ with no missing labels, adjacencies or branches such that $s_l(H) = s_l(H')$ and $r_l(H) = r(H')$. $H'$ is converted to simple history with the same cost as follows. 
\begin{itemize}
\item On every branch of $H'$ interpolate a vertex.  
\item Label each interpolated vertex identically to its parent. 
\item Connect the sides of the interpolated vertices to one another such that for any adjacency $\{ x_\alpha, y_\beta \}$ connecting interpolated vertices, $\{ A(x_\alpha), A(y_\beta) \} \in E_{H'}$. 
\end{itemize}
It is easily verified that the result is an AVG that can be edge partitioned into rearrangement and replication epochs and hence is a simple history.
\end{proof}

\subsubsection*{LBSC and LBRC are Lower Bounds}

\begin{lemma}
LBSC is a lower bound on substitution cost.
\label{LBSCLemma}

\end{lemma}

\begin{proof}
From Lemmas \ref{dnaToAVGLemma} and \ref{avgToRealisation} it follows that every history graph has a realisation. It is sufficient therefore to further prove that for any simple history $\textbf{H}$, $s(\textbf{H}) = s_l(\textbf{H})$  and that a history graph $G$ has no extension $G'$ such that $s_l(G) > s_l(G')$. The former is easily verified and we now prove the latter.

Let $(G = G_n) \red G_{n-1} \red \ldots G_2 \red (G_1 = G')$ be a sequence of $n$ history graphs for a reduction sequence of $n-1$ reduction operations.
For some integer $i \in [1, n)$ if the $i$th reduction operation is a vertex deletion, adjacency deletion or branch contraction, as these each have no impact on the calculation of LBSC, $s_l(G_{i+1}) = s_l(G_i)$.
Else the $i$th reduction operation is a label deletion. Let $x$ be the vertex whose label is being deleted. As the number of non-trivial lifted labels for $A(x)$ after the deletion of $x$ is less than or equal to the sum of non-trivial lifted labels for $x$ and $A(x)$, it follows that $s_l(G_{i+1}) \le s_l(G_i)$. Therefore by induction $s_l(G) \le s_l(G')$.
\end{proof}

\begin{lemma}
\label{LBRCLemma}
LBRC is a lower bound on rearrangement cost.
\end{lemma}

\begin{proof}
Analogously to the proof of Lemma \ref{LBSCLemma}, from Lemmas \ref{dnaToAVGLemma} and \ref{avgToRealisation} it follows that every history graph has a realisation. It is sufficient therefore to further prove that for any simple history $\textbf{H}$, $r(\textbf{H}) = r_l(\textbf{H})$  and that a history graph $G$ has no extension $G'$ such that $r_l(G) > r_l(G')$. The former is easily verified and we now prove the latter.

Let $(G = G_n)  \red G_n-1 \red ... G_2 \red (G_1 = G')$ be a sequence of $n$ history graphs for a reduction sequence of $n-1$ reduction operations. For some integer  $i \in [1, n)$ if the $i$th reduction operation is a label deletion, vertex deletion or contraction of a branch with a free-parent, as each removes an element that has no effect on the calculation of the LBRC, $r_l(G_{i+1}) = r_l(G_{i})$.

Else if the $i$th reduction operation is a contraction of a branch with a free-child, as the child is unattached the only possible effect on the LBRC calculation is the conversion of non-trivial lifted adjacencies into trivial lifted adjacencies, therefore $r_l(G_{i+1}) \le r_l(G_i)$ (see Figure \ref{appendixLowerBound}(A)).

Let $q(M)$ and $p(M)$ be the number of unattached and attached sides in a module $M$, as $q(M)+p(M) = V_M$:

\[ r_l(G) = \sum_{M \in M(G)} \lceil (q(M) + p(M))   / 2 \rceil - 1. \]

As each side may be incident with at most one adjacency $p(M)$ is even and $p(M)/2$ is the number of adjacencies in $M$, therefore:

\[ r_l(G) = | E_G | + \sum_{M \in M(G)} \lceil q(M) / 2 \rceil - 1. \]

Hence $r_l(G) = | E_G | + Q(G) - |M(G)|$, where $Q(G) = \sum_{M \in M(G)} \lceil q(M) / 2 \rceil$. Suppose $r_l(G_{i+1}) > r_l(G_i)$. If the $i$th reduction operation is an adjacency deletion, $|E_{G_{i+1}}| + 1 = |E_{G_{i}}|$, therefore $Q(G_{i+1}) - |M(G_{i+1})| \ge Q(G_i) - |M(G_i)| + 2$. 

\begin{figure}[h!]
\begin{center}
\includegraphics[width=12cm]{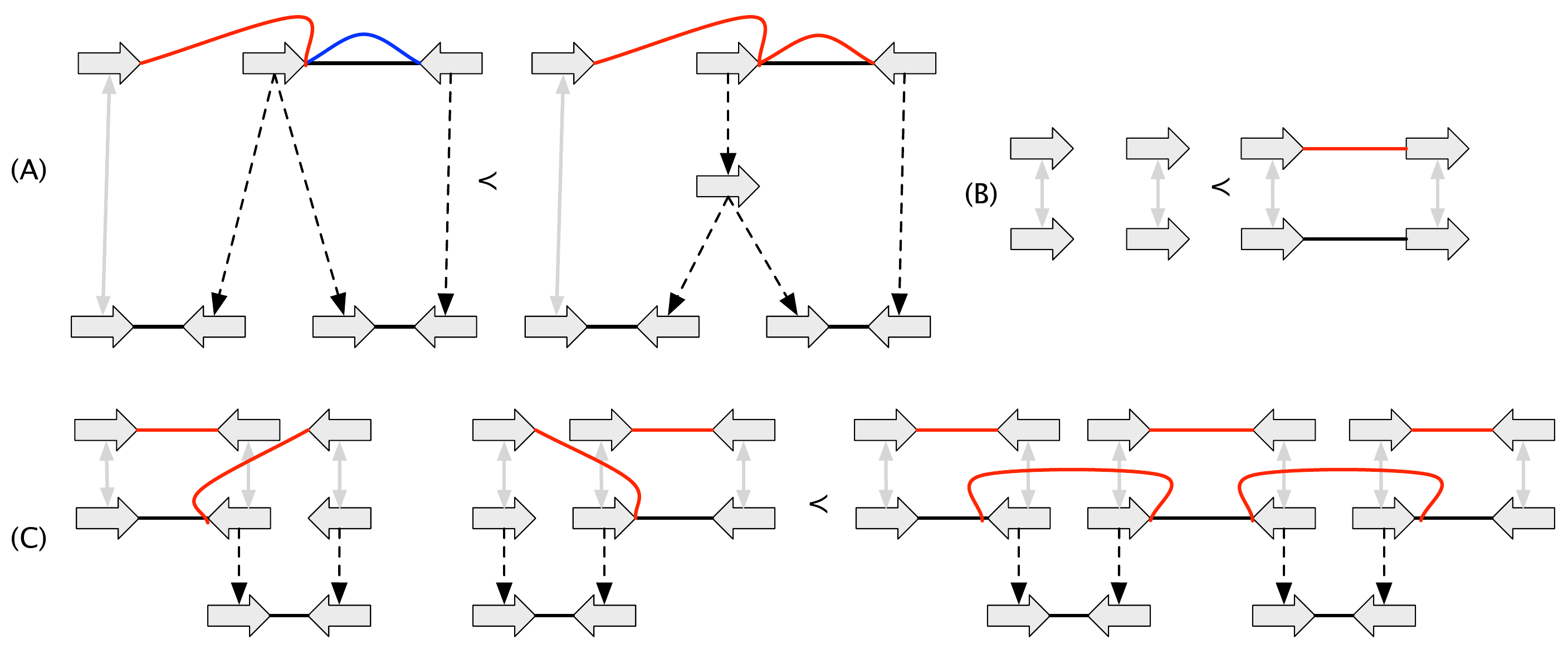}
\caption{(A) A contraction of a branch with a free-child can only possibly result in non-trivial adjacencies becoming trivial.  (B) A adjacency deletion can at most reduce the number of modules by 2, and if the number of modules decrease by two then the number of unattached sides in a modules decreases by 2. (C) An example of the deletion of an adjacency redistributing two unattached sides.}
\label{appendixLowerBound}
\end{center}
\end{figure}

The removal of an adjacency can reduce the number of modules by at most two, therefore $|M(G_i)| - |M(G_{i+1})| \le 2$. The number of modules decreases by the maximum of two only when the adjacency to be deleted connects two sides that each have no incident lifted adjacencies (see Figure \ref{appendixLowerBound}(B)). However, in this case $Q(G_i) = Q(G_{i+1}) + 1$, as the number of unattached sides in a module decreases by $2$, therefore if $|M(G_i)| - |M(G_{i+1})| = 2$ then $r_l(G_{i+1}) \le r_l(G_i)$. 

An unattached side in a module is the side of a free-root, and such a free-root side has incident lifted adjacencies. The side of a free-root with no incident lifted adjacencies can not become part of a module by the removal of any adjacency from the associated history graph, as by definition the homologous sides in its associated branch-tree are all unattached. The removal of an adjacency can therefore only decrease or leave the same the total number of unattached sides in modules. The only way for $Q(G_{i+1}) - Q(G_i)$ to be positive is therefore by the redistribution of unattached sides between modules to exploit the ceiling function. As in the removal of a single adjacency at most two unattached sides can be redistributed from a single module (see Figure \ref{appendixLowerBound}(C)), therefore $Q(G_{i+1}) - Q(G_i) \le 1$. But if $Q(G_{i+1}) - Q(G_i) = 1$ then it is easily verified $|M(G_i)| - |M(G_{i+1})| \le 0$.
This is all the cases, therefore $r_l(G_{i+1}) \le r_l(G_i)$, by induction therefore $r_l(G) \le r_l(G')$.
\end{proof}


\newtheorem*{theorem1}{Theorem 1}

\begin{theorem1}
For any history graph $G$ and any cost function $c$, $c(s_l(G), r_l(G)) \le C(G, c) \le c(s_u(G), r_u(G))$ with equality if $G$ is an AVG.
\end{theorem1}

\begin{proof}
Follows from Lemmas  \ref{equivalentSubstitutionBoundsLemma}, \ref{equivalentRearrangementBoundsLemma}, \ref{dnaToAVGLemma}, \ref{avgToRealisation}, \ref{LBSCLemma} and \ref{LBRCLemma}.
\end{proof}

\subsection{Appendix B}

In this section we will prove Theorem \ref{containsGOptimal}. Towards this aim we classify non-minimal adjacencies and labels.



A non-minimal label of a vertex $x$ is (see Figure \ref{nonMinimalElementsClassification}(A)):
\begin{itemize}
\item A \emph{leaf} if $L'_{x} = \{ \}$, 
\item else, as it is not a junction, $|L'_{x}| = 1$ and:
\begin{itemize} 
\item the label is \emph{redundant} if $\tilde{L}'_{x} = \{ \}$,
\item else \emph{complicating} if $l(A(x)) \not= l(x)$, 
\item else $l(A(x)) = l(x)$ and, as it is not a bridge, then $\tilde{L}_{A(x)} = \{ \}$ and it is an \emph{unnecessary bridge}.
\end{itemize}
\end{itemize}




A non-minimal adjacency $\{ x_{\alpha}, y_{\beta} \}$ is (see Figure \ref{nonMinimalElementsClassification}(B)):
\begin{easylist}[itemize]
& a \emph{leaf} if $L'_{x_{\alpha}} \cup L'_{y_{\beta}} = \{ \}$, 
& else,  as it is not a junction,  neither $x_{\alpha}$ or $y_{\beta}$ are junction sides and it is \emph{complex} if $| L'_{x_{\alpha}} | > 1$ or $| L'_{y_{\beta}} | > 1$, 
& else  $| L'_{x_{\alpha}} | \le 1$, $| L'_{y_{\beta}} | \le 1$ and:
&& the adjacency is \emph{redundant} if $L_{x_{\alpha}} \cup L_{y_{\beta}} = \{ \{ x_{\alpha}, y_{\beta} \} \}$,
&& else \emph{complicating} if $\{ A(x_{\alpha}), A(y_{\beta}) \}$ is a non-trivial lifted adjacency, 
&& else $\{ A(x_{\alpha}), A(y_{\beta}) \}$  is a trivial lifted adjacency and, as it is not a bridge either:
&&& $\tilde{L}'_{A(x_{\alpha})} \cup \tilde{L}'_{A(y_{\beta})} = \{ \}$ and it is an \emph{unnecessary bridge},
&&& else $(L'_{A(x_{\alpha})} \cup L'_{x_\alpha}) \setminus  \{ \{ x_\alpha, y_\beta \} \} \le 1$ and $(L'_{A(y_{\beta})} \cup L'_{y_\beta}) \setminus \{ \{ x_\alpha, y_\beta \} \} \le 1$ and it is a \emph{removable bridge}.
\end{easylist}



\begin{lemma}
\label{noNonMinimalElementsLemma}
A $G$-minimal AVG contains no $G$-reducible non-minimal elements.
\end{lemma}

\begin{proof}
We prove the contrapositive. It is easily verified that the deletion of any single non-minimal vertex or contraction of a non-minimal branch from an AVG results in a reduction that is also an AVG. It is also easily verified that the deletion of each possible type of non-minimal label/adjacency from an AVG results in a reduction that is also an AVG, with the exceptions of a complex non-minimal adjacency, which can not be present within an AVG (because such an edge implies ambiguity), and a removable bridge adjacency. After deletion of a removable bridge adjacency $\{ x_\alpha, y_\beta \}$ the adjacency $\{ A(x_\alpha), A(y_\beta) \}$ ceases to be a junction adjacency, and may either become a bridge, in which case the resulting graph is an AVG, or it may become a non-minimal adjacency. If it becomes a non-minimal adjacency, then, by the prior argument, if it is not a removable bridge adjacency then its deletion results in an AVG, else if it is a removable bridge then after the deletion of $\{ A(x_\alpha), A(y_\beta) \}$, the process of considering if $\{ A(A(x_\alpha)), A(A(y_\beta)) \}$ is non-minimal and deleting if necessary is repeated iteratively until the resulting graph is an AVG.
\end{proof}

\begin{figure}[h!]
\begin{center}
\includegraphics[width=10cm]{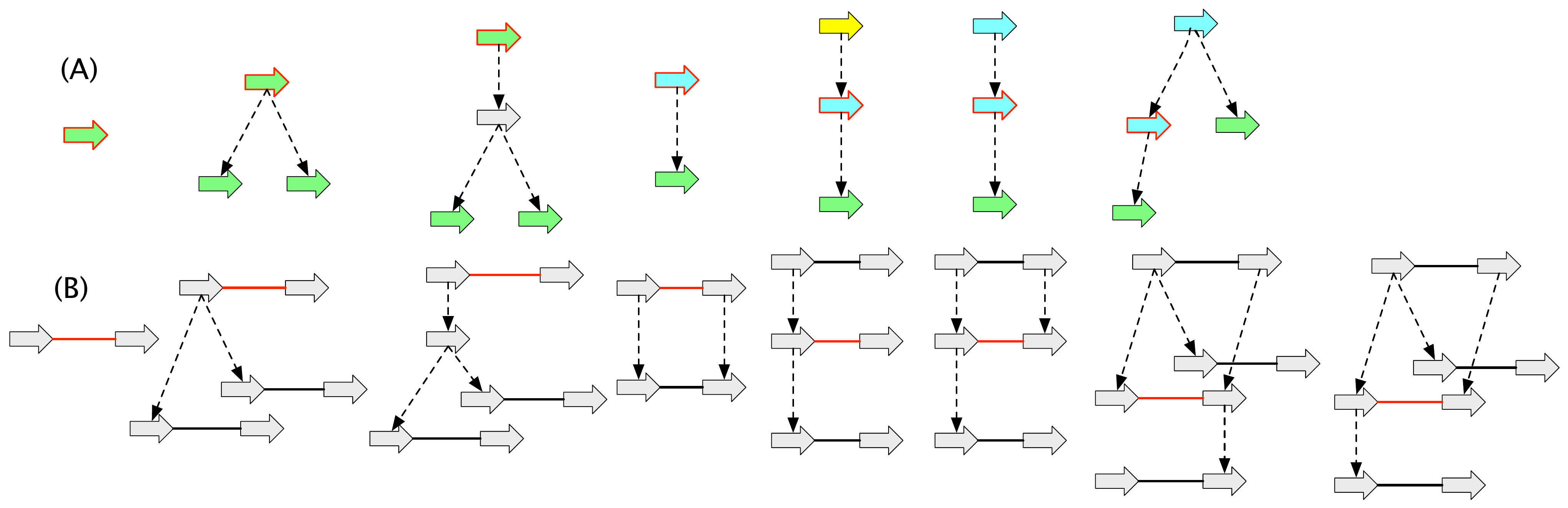}
\caption{(A) Classification of labels. From left-to-right labels of vertices outlined in red are: leaf, junction, (another) junction, redundant, complicating, unnecessary bridge and bridge. (B) Classification of adjacencies. From left-to-right adjacencies in red are: leaf, junction, complex, redundant, complicating, unnecessary bridge, removable bridge and bridge.} 
\label{nonMinimalElementsClassification}
\end{center}
\end{figure}

\begin{figure}[h!]
\begin{center}
\includegraphics[width=7cm]{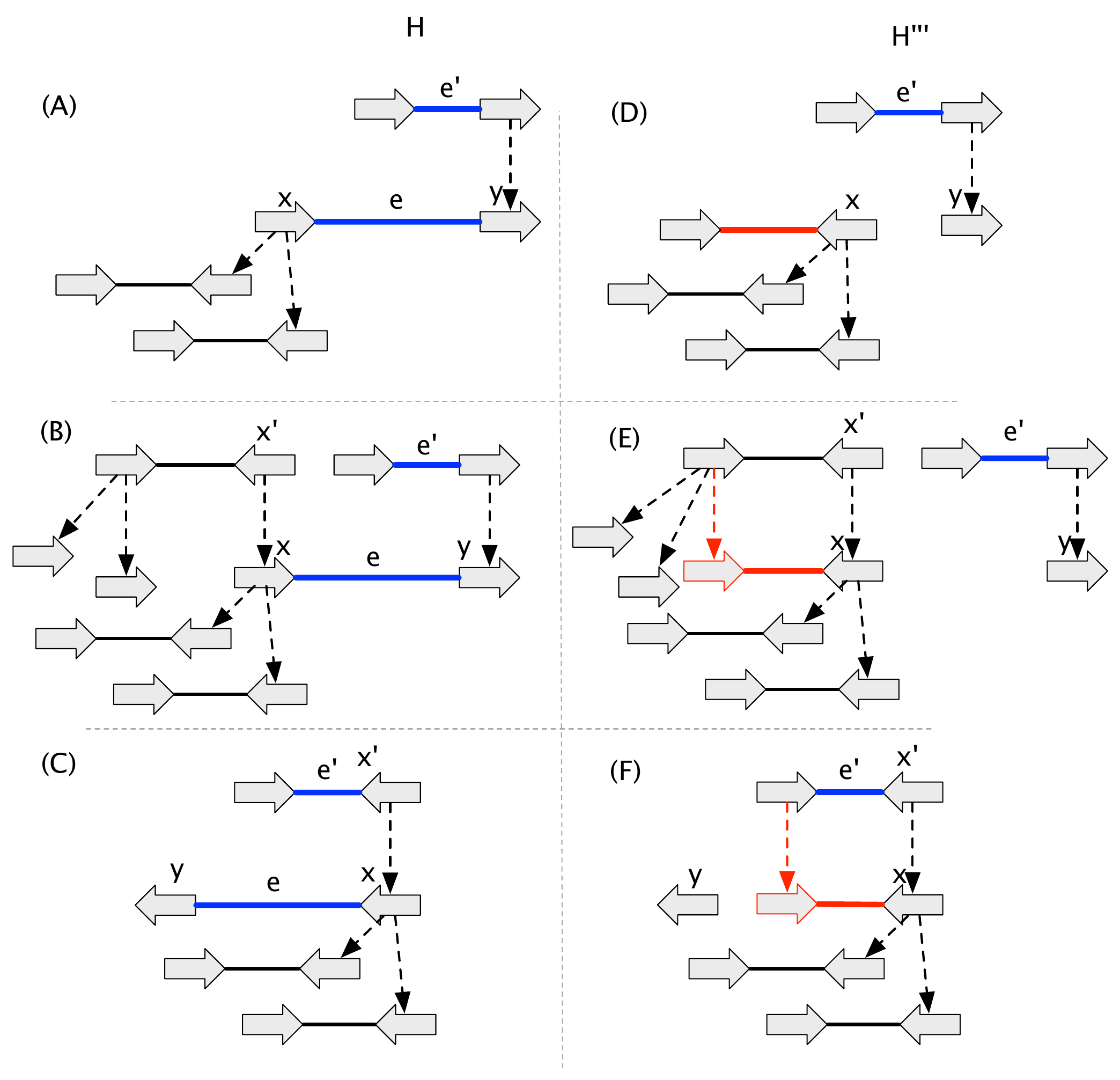}
\caption{\textbf{(A,B,C)} Examples of ping-pong adjacencies. The ping-pong adjacencies are shown in blue. \textbf{(D,E,F)} After modifications to remove the ping-pong adjacencies, for each corresponding left side case, with the added elements shown in red. }
\label{pingPongs}
\end{center}
\end{figure}

\begin{lemma}
\label{onlyJunctionAdjacenciesLemma}
The only $G$-reducible adjacencies in the $G$-unbridged graph of an extension of $G$ containing no non-minimal elements are junction adjacencies. 
\end{lemma}

\begin{proof}
By definition, the only $G$-reducible adjacencies in an extension of $G$ with no $G$-reducible non-minimal elements are junction adjacencies and bridges. Each deletion of a $G$-reducible  bridge adjacency does not create any $G$-reducible non-minimal adjacencies, as a junction adjacency connecting sides that are the lifting ancestors of the sides connected by a bridge adjacency remains a junction adjacency after the deletion of the bridge, and the lifted adjacencies incident with the sides connected by the bridge, which are non-trivial, lift to this junction instead and therefore remain non-trivial.
\end{proof}

\begin{lemma}
The $G$-unbridged graph of a $G$-optimal AVG for any cost function contains no $G$-reducible ping adjacencies.
\label{noPingsElementsLemma}
\end{lemma}

\begin{proof}
Let $H$ be a $G$-optimal AVG whose $G$-unbridged graph $H'$ contains one or more $G$-reducible ping adjacencies. Example subgraphs containing ping-pong adjacencies are shown in Figure \ref{pingPongs}(A-C). Let $e = \{ x_{\alpha}, y_{\beta} \}$ be such a $G$-reducible ping adjacency and $y_{\beta}$ a hanging endpoint in $H'$. From Lemma \ref{onlyJunctionAdjacenciesLemma}, the adjacency $e$ must be a junction. Delete $e$ from $H'$ giving $H''$, note $G \redeq H''$. If $x_{\alpha}$ has no most recent attached ancestor create a new vertex and connect it with an adjacency to $x_\alpha$ as shown in Figure \ref{pingPongs}(D), else do the same but connect the new vertex by a branch that makes it the child of the vertex connected by an adjacency to the most recent attached ancestor of $x_\alpha$, as shown in Figure \ref{pingPongs}(E). Note that it does not matter in this second case if the most recent attached ancestor of $x_\alpha$ is a pong adjacency, as demonstrated in Figures \ref{pingPongs}(C) and (F). It is easily verified that each modification defines the $G$-unbridged graph of a valid AVG extension $H'''$ of $G$ that has one fewer $G$-reducible ping adjacencies in its $G$-unbridged graph, one less rearrangement and the same number of substitutions in its most parsimonious realisation as in the most parsimonious realisation of $H$. This contradiction to the assumption that $H$ was $G$-optimal establishes the result.
\end{proof}

\newtheorem*{theorem2}{Theorem 2}

\begin{theorem2}
The $G$-bounded AVGs contain the $G$-optimal AVGs for every cost function.
\end{theorem2}

\begin{proof}
Follows from Lemmas \ref{noNonMinimalElementsLemma} and \ref{noPingsElementsLemma}.
\end{proof}

\subsection{Appendix C}

This section will prove Theorem \ref{finiteGBoundedTheorem}.
In the following let $n$ be the number of adjacencies in a history graph $G$. 

 \begin{lemma} \label{nIs0Lemma}
 If $n=0$ any $G$-bounded extension of $G$ contains 0 adjacencies.
 \end{lemma}
 \begin{proof}
 Follows from Lemma \ref{noNonMinimalElementsLemma}.
 \end{proof}
 
 As the $n=0$ case is trivial now assume that $n \ge 1$.
For an adjacency $\{ x_{\alpha}, y_{\beta} \}$ its \emph{received incidence} is $|L'_{x_{\alpha}}| + | L'_{y_{\beta}} |$ and its \emph{projected incidence} is equal to the number of members of $\{ A(x_{\alpha}), A(y_{\beta}) \}$ that are attached, either 0, 1 or 2. For an adjacency, the difference between projected incidence and received incidence is the \emph{incidence transmission}.
A positive incidence transmission occurs when the projected incidence is greater than the received incidence number, conversely a negative incidence transmission occurs when the projected incidence is less than the received incidence.
The \emph{incidence sum} of a history graph is the sum of the received incidences of its adjacencies, or, equivalently, the sum of the projected incidences of its adjacencies.



\begin{lemma}
\label{twoNMinusTwoLemma}
The maximum possible incidence sum of $G$ is $2n - 2$.
\end{lemma}

\begin{proof}
The $2n$ term is because each adjacency has a projected incidence of at most $2$,  the $-2$ term is because at least one adjacency has a projected incidence of 0. 
\end{proof}

It is trivial to show this bound can be achieved for all values of $n$. 

\begin{lemma} \label{noPositiveJunctionsLemma}
The $G$-unbridged graph $G''$ for a $G$-bounded history graph $G'$ has no $G$-reducible adjacencies with a positive incidence transmission.
\end{lemma}

\begin{proof}
By Lemma \ref{onlyJunctionAdjacenciesLemma}, the only $G$-reducible adjacencies in $G''$ are junction adjacencies. Junction adjacencies have an incidence transmission of 0 or less. 
\end{proof}

\begin{lemma}
\label{twoNMinusOneLemma}
The $G$-unbridged graph $G''$ for a $G$-bounded history graph $G'$ contains less than or equal to $2n - 1$ adjacencies that either have a negative incidence transmission, or which are $G$-irreducible and have an incidence transmission of 0.
\end{lemma}

\begin{proof}
Let $k_{i,j}$ be the number of adjacencies in $G''$ that have a projected incidence of $i$ and a received incidence of $j$. As the sum of projected incidences equals the sum of received incidences therefore:

\[ \sum_{i=0}^{2} \sum_{j=0}^\infty i \mbox{ } k_{i,j} = \sum_{i=0}^{2} \sum_{j=0}^\infty j \mbox{ } k_{i,j}\]

Separating the contributions of adjacencies with a negative incidence transmission:

\[ \sum_{i=0}^{2} \sum_{j=0}^i i \mbox{ } k_{i,j} + \sum_{i=0}^{2} \sum_{j=i+1}^\infty i \mbox{ } k_{i,j} = \sum_{i=0}^{2} \sum_{j=0}^i j \mbox{ } k_{i,j} + \sum_{i=0}^{2} \sum_{j=i+1}^\infty j \mbox{ } k_{i,j},\]

\[ \sum_{i=0}^{2} \sum_{j=0}^{i-1} (i - j) \mbox{ } k_{i,j}  = \sum_{i=0}^{2} \sum_{j=i+1}^\infty (j - i) \mbox{ } k_{i,j},  \]

\[ \sum_{i=0}^{2} \sum_{j=0}^{i-1} (i - j) \mbox{ } k_{i,j} - \sum_{i=0}^{2} \sum_{j=i+1}^\infty (j - i - 1) \mbox{ } k_{i,j} = \sum_{i=0}^{2} \sum_{j=i+1}^\infty  k_{i,j}.  \]

The first term of the left-hand side of the equation is the total incidence transmission of all adjacencies in $G''$ with a positive incidence transmission. Using Lemma \ref{noPositiveJunctionsLemma}, these adjacencies must all be $G$-irreducible. Let $k$ be the number of $G''$-irreducible adjacencies that have an incidence transmission of 0, as: \[ \sum_{i=0}^{2} \sum_{j=0}^{i-1} k_{i,j} < n - k, \] \[ \sum_{i=0}^{2} \sum_{j=0}^{i-1} (i - j) \mbox{ } k_{i,j} \leq 2 \sum_{i=0}^{2} \sum_{j=0}^{i-1} k_{i,j} \leq 2n - 2k\] therefore by substitution:
  
  \[ 2n - \sum_{i=0}^{2} \sum_{j=i+1}^\infty (j - i - 1) \mbox{ } k_{i,j} \geq \sum_{i=0}^{2} \sum_{j=i+1}^\infty  k_{i,j} + 2k  \]
  
  The right-hand side of the inequality is the number of adjacencies with a negative incidence transmission plus two times the number of $G$-irreducible adjacencies with an incidence transmission of 0.
  
 As $ \sum_{i=0}^{2} \sum_{j=i+1}^\infty (j - i - 1) k_{i,j}$ can not  be negative, it remains only to prove that this term must be positive. Assume that there are $2n$ or more adjacencies that either have a negative incidence transmission, or which are $G$-irreducible and have an incidence transmission of 0 (i.e. a contradiction of the lemma). As $n > 0$, there must be at least one adjacency in $G''$ with a projected incidence of 0 and a received incidence of greater than 0 in some ancestral thread (i.e. $\sum_{j=1}^\infty k_{0,j} > 0$). Either such an edge has a received incidence of 2 or greater, in which case the considered term must be positive, or a larger graph exists (see Figure \ref{incidenceTransmissionLemma2}) that is $H$-bounded and $H$-unbridged with respect to a graph $H$, which has the same number of adjacencies as $G$ and an extra edge with a projected incidence of 0 and a received incidence of 2 or greater, which implies that $\sum_{i=0}^{2} \sum_{j=i+1}^\infty  k_{i,j} + 2k < 2n$. In either case we derive a contradiction to the assumption of the number of adjacencies, therefore:
  
  \[ 2n - 1 \ge  \sum_{i=0}^{2} \sum_{j=i+1}^\infty k_{i,j} + 2k. \]
  
 \end{proof}
 
 \begin{figure}[h!]
\begin{center}
\includegraphics[width=3cm]{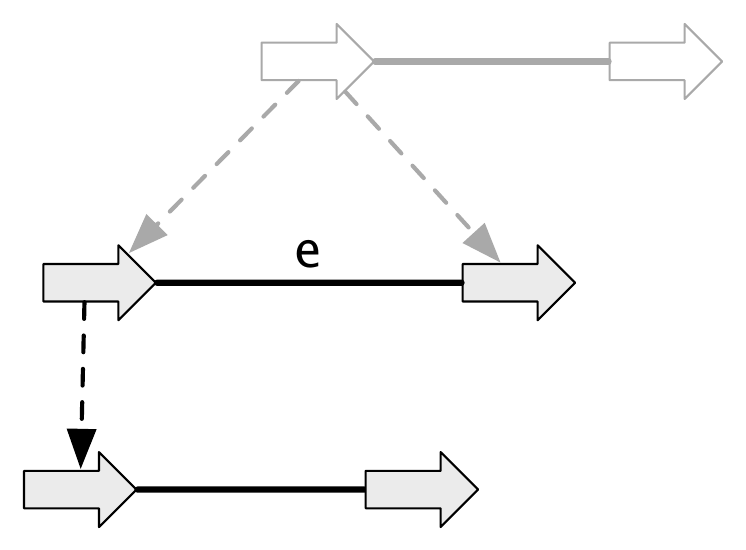}
\caption{Ignoring the grey elements, if there exists an adjacency $e$ in a $G$-bounded and $G$-unbridged graph with no projected incidences and one received incidence, it must be $G$-irreducible and a larger graph, shown by the elements in grey, exists that is $H$-bounded for a graph $H$ with the same number of adjacencies as $G$.}
\label{incidenceTransmissionLemma2}
\end{center}
\end{figure}



 
 
    
  
  

\begin{lemma}
\label{adjacencyCountingLemma}
The $G$-unbridged graph $G''$ of a $G$-bounded history graph $G'$ contains less than or equal to $3n - 3$ $G$-reducible adjacencies with an incidence transmission of 0. 
\end{lemma}

\begin{proof}
Let $X$ be the set of $G$-reducible adjacencies with an incidence transmission of 0 in $G''$. The sum of received incidences equals the sum of projected incidences for members of $X$, therefore the sum of received incidences of other adjacencies in $G''$ ($\le 2n-1$ $G$-reducible adjacencies with positive or negative incidence transmission by Lemmas \ref{noPositiveJunctionsLemma} and \ref{twoNMinusOneLemma} and $\le n$ $G$-irreducible adjacencies) is equal to the sum of their projected incidences, which by Lemma \ref{twoNMinusTwoLemma} is at most $2(n + 2n - 1) - 2 = 6n - 4$. By Lemma \ref{onlyJunctionAdjacenciesLemma}, any adjacency $e=\{ x_{\alpha}, y_{\beta} \}$ in $X$ must be a junction adjacency, and, as it has 0 incidence transmission, must have a hanging endpoint and projected incidence of 2 (see Figure \ref{incidenceTransmissionLemma}(A)). Let $e'$ and $e''$ be the adjacencies incident with $A(x_{\alpha})$ and $A(y_{\beta})$, respectively (see Figure \ref{incidenceTransmissionLemma}(A)). As there exist no $G$-reducible ping adjacencies, $e'$ and $e''$ are either $G$-irreducible or $G$-reducible junction adjacencies with a negative incidence transmission. As $e$ projects at least one incidence to each such adjacency, $X$ has a cardinality at most $(6n - 4)/2 = 3n - 2$. It remains to prove that it must be at least one less than this bound.

Now let $e$ be a $G$-reducible junction adjacency in $G''$ that is contained in a thread that is ancestral or unrelated to all threads that contain a $G$-reducible adjacency or label. If $G''$ contains more adjacencies than $G$ then such an adjacency must clearly exist in $G''$. 

If $e$ makes projected incidences to $G$-irreducible adjacencies then it makes projected incidences to adjacencies not in $X$.
If $e$ does not make projected incidences then it has negative incidence transmission, and either $e$ is a hanging adjacency, in which case it must receive projected incidences from adjacencies that are not in $X$ (else there exists a $G$-reducible ping adjacency), or $e$ is not a hanging adjacency and a larger graph exists (see Figure \ref{incidenceTransmissionLemma}(B)) that is $H$-bounded with respect to a graph $H$ with the same number of adjacencies as $G$, in which case, using Lemma \ref{twoNMinusOneLemma}, there must be less than $2n - 1$ $G$-reducible negative transmission incidence adjacencies in $G''$. Therefore either there exist projected incidences made between adjacencies not in $X$ or there are fewer than $2n - 1$ $G$-reducible negative transmission incidence adjacencies in $G''$, either way, there are fewer than $6n - 4$ projections made from adjacencies in $X$ to adjacencies not in $X$, and as there are no projections made between adjacencies in $X$, and all adjacencies in $X$ have a projected incidence of 2, therefore $X$ has cardinality less than $3n - 2$.    
\end{proof}

\begin{figure}[h!]
\begin{center}
\includegraphics[width=5cm]{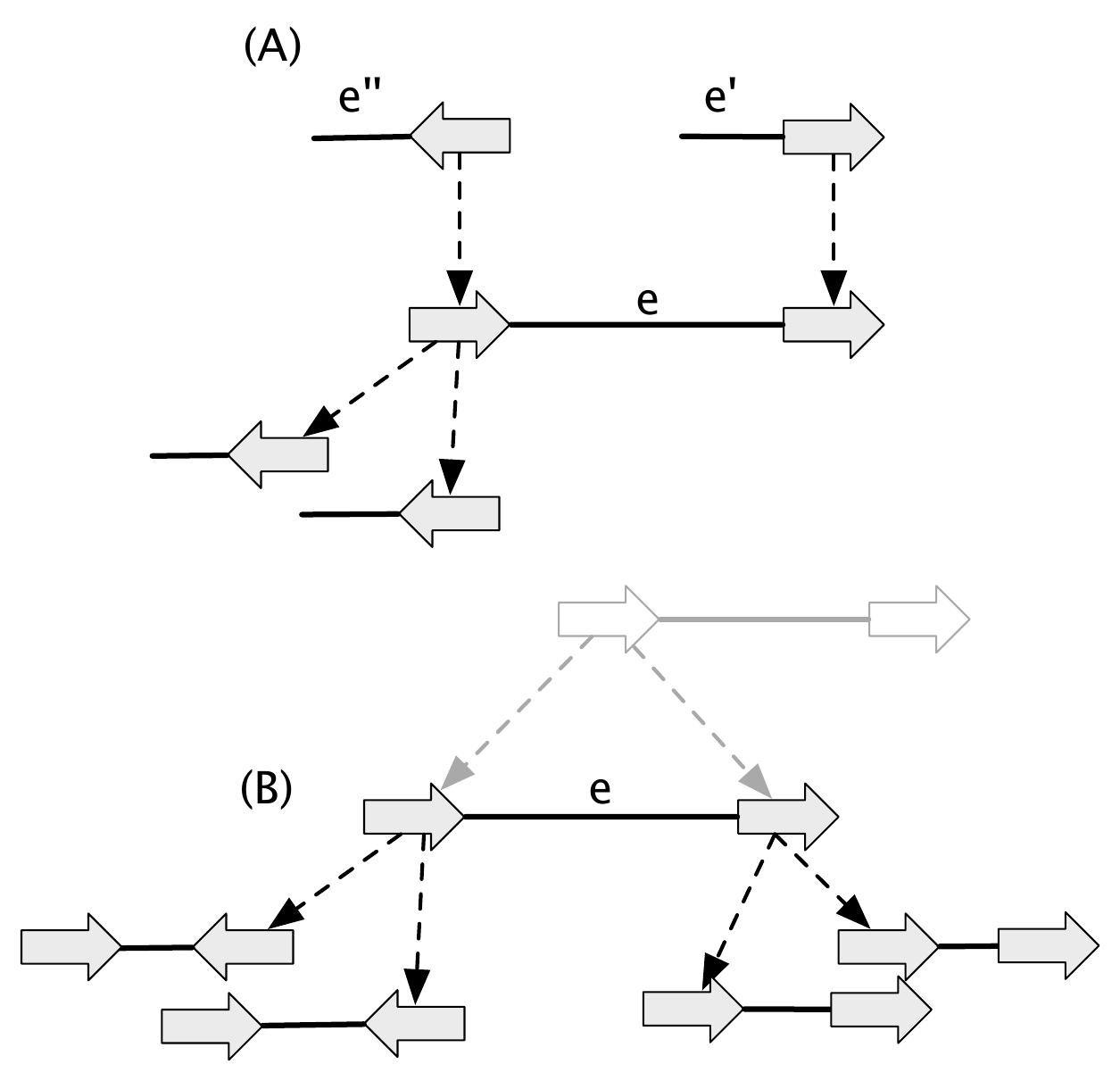}
\caption{\textbf{(A)} A junction adjacency with incidence transmission of 0 in a $G$-bounded AVG. \textbf{(B)} Ignoring the grey elements, if the adjacency $e$ in a $G$-bounded graph is a $G$-reducible junction adjacency with no projected incidences and no hanging endpoints, a larger graph, shown by the elements in grey, exists that is $H$-bounded for a graph $H$ with the same number of adjacencies as $G$.}
\label{incidenceTransmissionLemma}
\end{center}
\end{figure}

\begin{lemma}
\label{junctionAdjacenciesBoundLemma}
A $G$-bounded history graph $G'$ contains less than or equal to $5n - 4$ junction adjacencies.
\end{lemma}

\begin{proof}
From Lemmas \ref{noPositiveJunctionsLemma}, \ref{twoNMinusOneLemma} and \ref{adjacencyCountingLemma} it follows that the unbridged graph of $G'$ contains less than $5n - 4$ junction adjacencies. Extending the argument of Lemma \ref{onlyJunctionAdjacenciesLemma}, it is easily verified that $G'$ contains the same number of junction adjacencies as its unbridged graph. 
\end{proof}

\begin{lemma}
\label{adjacenciesBoundLemma}
A $G$-bounded history graph $G'$ contains less than or equal to $10n - 8$ $G$-reducible adjacencies and $20n - 16$ additional attached vertices. These bounds are tight for all $n \ge 1$.
\end{lemma}

\begin{proof}
Let $i$ and $j$ be the numbers of $G$-reducible junction and bridge adjacencies in $G'$, respectively. As bridges and junctions are the only $G$-reducible adjacencies in $G'$, $i + j$ is equal to the total number of $G$-reducible adjacencies in $G'$. Assume that $i + j > 10n - 8$. From Lemma \ref{junctionAdjacenciesBoundLemma} it follows that $i \le 5n - 4$, therefore $j > 5n - 4$. As $j > 5n - 4$, it follows from Lemma \ref{junctionAdjacenciesBoundLemma} there exists in $G'$ a pair of $G$-reducible bridge adjacencies $\{ x_\alpha, y_\beta \}$, $\{ w_\alpha, z_\beta \}$ such that   $\{ A(x_\alpha), A(y_\beta) \} = \{ A(w_\alpha), A(z_\beta) \}$ (see Figure \ref{lessBridgesThanJunctions}(A)). However, in this case there exists  
an extension $H$ of $G$ that contains the same number of adjacencies as $G'$ but one additional $G$-reducible junction adjacency (see Figure \ref{lessBridgesThanJunctions}(B)), therefore in $H$ the number of $G$-reducible junction adjacencies is greater than $5n-4$, a contradiction of Lemma \ref{junctionAdjacenciesBoundLemma}, therefore $i + j \le 10n-8$. From this bound, trivially, the bound of the number of additional attached vertices follows. Figure \ref{sizeOfGBoundedGraph1} shows both bounds are tight for all $n$.
\end{proof}

\begin{figure}[h!]
\begin{center}
\includegraphics[width=9cm]{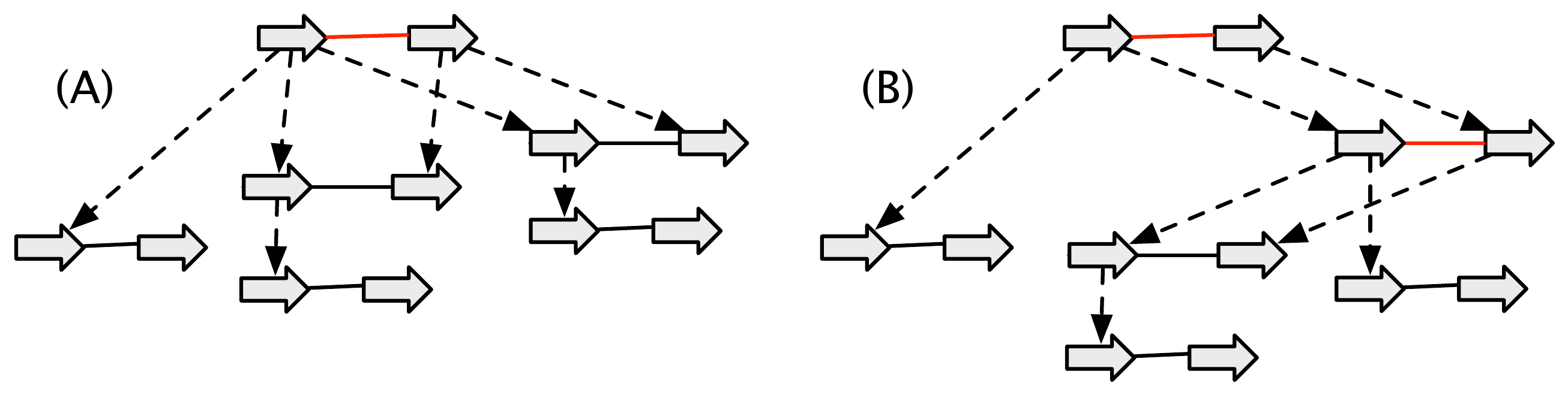}
\caption{\textbf{(A)} A graph with a single junction adjacency (coloured red) and two bridge adjacencies. \textbf{(B)} A graph with the same size and cardinality as that in (A), but with an additional junction adjacency.}
\label{lessBridgesThanJunctions}
\end{center}
\end{figure}

\begin{figure}[h!]
\begin{center}
\includegraphics[width=6cm]{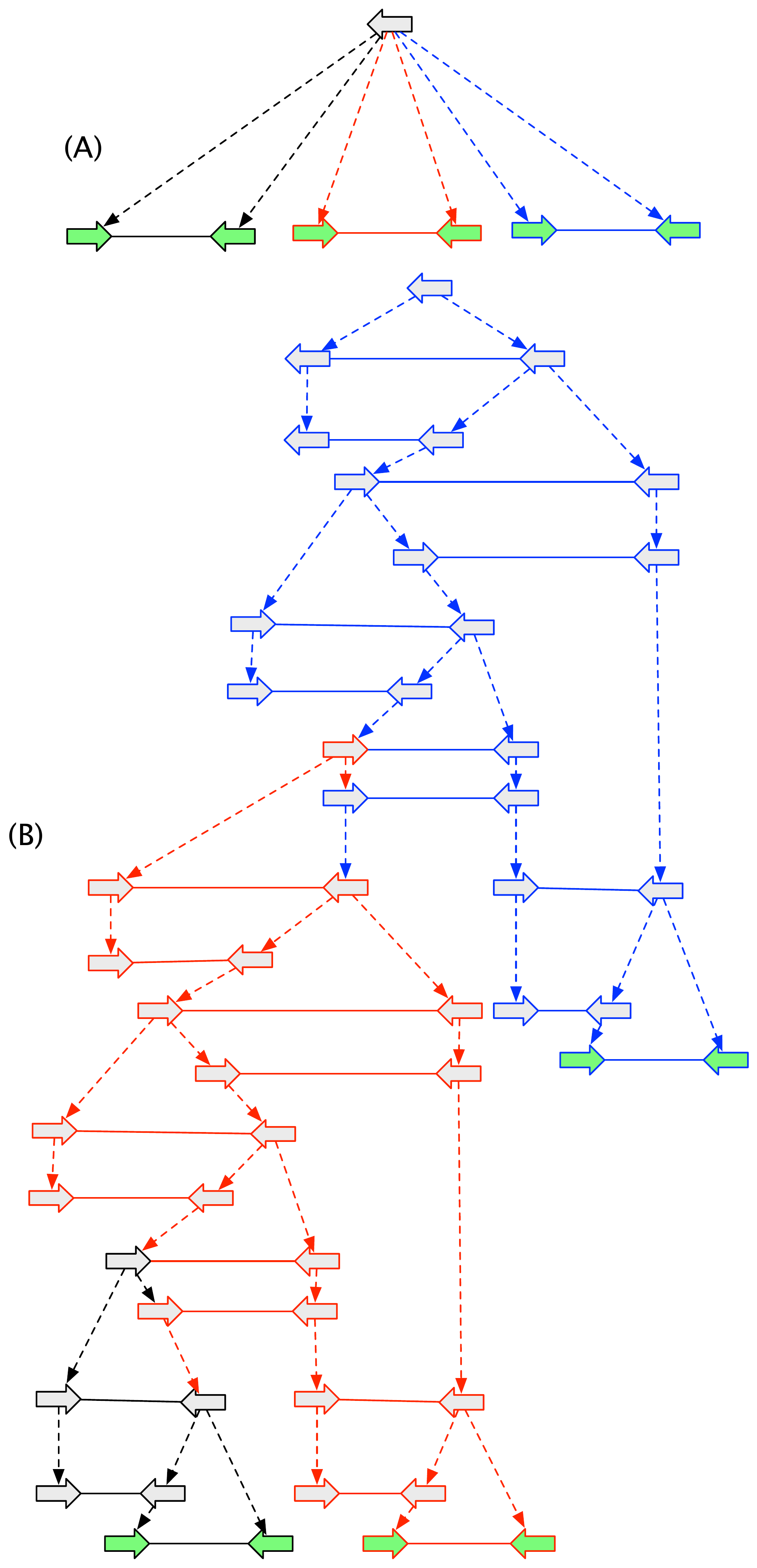}
\caption{\textbf{(A)} A history graph with 3 adjacencies. \textbf{(B)} A $G$-bounded AVG extension of the graph in (A) with 22 extra $G$-reducible adjacencies and 44 extra attached vertices, the maximum possible. The leaf adjacencies connecting the labeled leaves are $G$-irreducible, as there are no other labeled vertices in $G'$. All the other adjacencies are junctions or bridges.
The number of elements in the red subgraph of (B) corresponds to the maximum number of adjacencies and attached vertices  that can be added given the inclusion of the red subgraph in (A), and similarly for the blue subgraph. By extrapolation this demonstrates the bounds on the number of additional adjacencies and attached vertices are tight for all possible numbers of adjacencies in the original graph.}
\label{sizeOfGBoundedGraph1}
\end{center}
\end{figure}


Let $m$ be the number of labeled vertices in the history graph $G$.   As with the $n=0$ case, the $m=0$ case is similarly trivial, but in terms of the number of $G$-reducible labels. 

  \begin{lemma} \label{mIs0Lemma} 
 If $m=0$ any $G$-bounded extension of $G$ contains 0 labels.
 \end{lemma}
 \begin{proof}
 Follows from Lemma \ref{noNonMinimalElementsLemma}.
\end{proof}
 
Now assume that $m \ge 1$ and that $n \ge 0$.

\begin{lemma}
A $G$-bounded history graph $G'$ contains less than or equal to $2m - 2$ $G$-reducible vertex labels. This bound is tight for all $m \ge 1$.
\label{labelsBoundLemma}
\end{lemma}

\begin{proof}
Let $i$ and $j$ be the number of junction and bridge labels, respectively, in $G'$. By Lemma \ref{noNonMinimalElementsLemma}, the total number of $G$-reducible labels in $G'$ is less than or equal to $i + j$.
The number of bridges $j$ is less than or equal to the total number of child branches of vertices that are label junctions, which, as the connected components of branches are trees, is equal or fewer than two times the number of leaf labels minus 2, and therefore equal or fewer than $2m - 2$. Furthermore, by definition, the lifting ancestor of a vertex with a bridge label has a non-trivial lifted label, which implies such a vertex's label is not a bridge, therefore $j \le 2m - 2 - i$, therefore $j + i \le 2m - 2$. Figure \ref{sizeOfGBoundedGraph2} shows this bound is tight for all $m$.
\end{proof}

\begin{figure}[h!]
\begin{center}
\includegraphics[width=7cm]{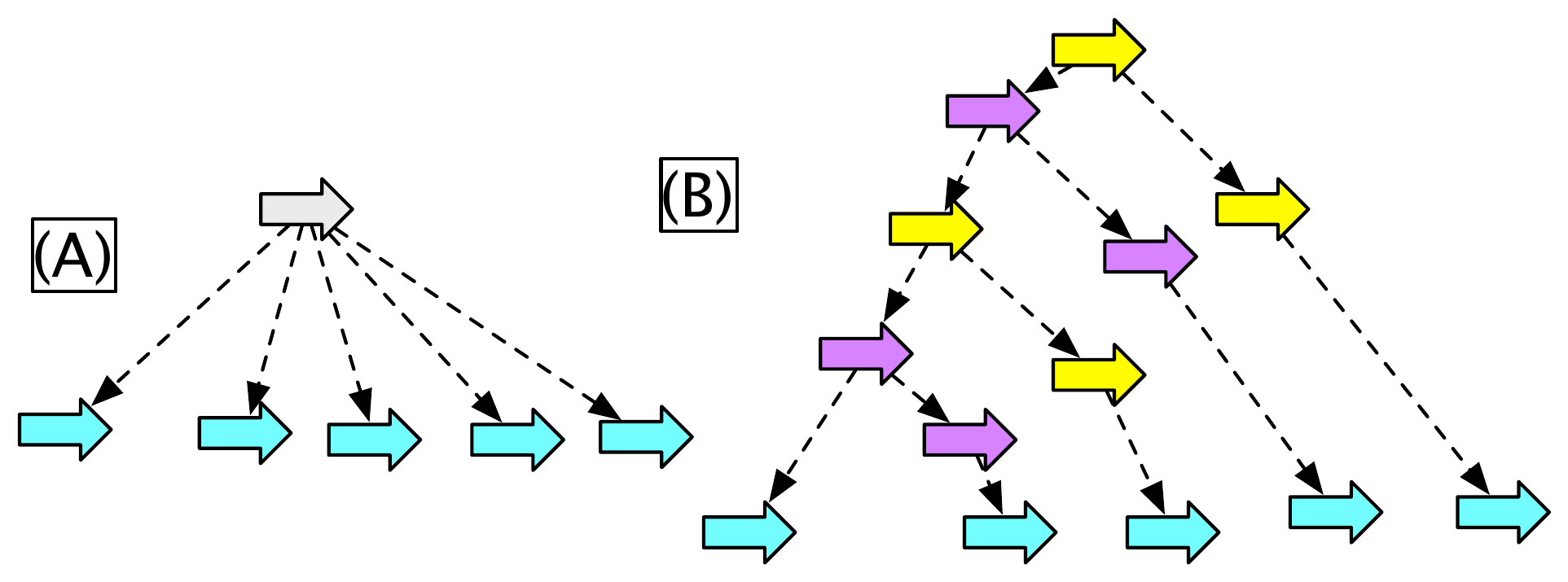}
\caption{\textbf{(A)} A history graph with 5 labeled vertices. \textbf{(B)} A $G$-bounded AVG extension of the graph in (A) with 8 extra labeled vertices, the maximum possible. As only the leaf labels have the given label colour their labels are $G$-irreducible, all the other labels in the graph are bridges or junctions. For each extra labeled leaf vertex added to (A) an extra pair of junction and bridge adjacencies can be added to (B), thus the bound is tight for all $m$.}
\label{sizeOfGBoundedGraph2}
\end{center}
\end{figure}


We are now in a position to prove the desired theorem for any value of $n$ and $m$.

\newtheorem*{theorem3}{Theorem 3}

\begin{theorem3} 
A $G$-bounded history graph contains less than or equal to $\max(0, 10n - 8)$ $G$-reducible adjacencies and $\max(0, 2m - 2, 20n - 16, 20n + 2m - 18)$ additional vertices. This bound is tight for all values of $n \ge 0$ and $m \ge 0$.
\end{theorem3}

\begin{proof}
Lemmas and \ref{nIs0Lemma} and \ref{adjacenciesBoundLemma} prove the bound on the number of $G$-reducible adjacencies, it remains to prove the bound on the number of additional vertices.  

Let $X$, $Y$ and $Z$ be the total numbers, respectively, of additional attached, labeled and both unattached and unlabeled vertices in $G'$.

From Lemmas \ref{nIs0Lemma} and \ref{adjacenciesBoundLemma} it follows that $X \le \max(0, 20n - 16)$. From Lemmas \ref{mIs0Lemma} and \ref{labelsBoundLemma} it follows that $Y \le \max(0, 2m - 2)$. Combining these results $X + Y \le \max(0, 2m - 2, 20n - 16, 20n + 2m - 18)$. 

Assume $X + Y + Z > \max(0, 2m - 2, 20n - 16, 20n + 2m - 18)$. As $X + Y \le \max(0, 2m - 2, 20n - 16, 20n + 2m - 18)$, $Z \ge 1$.
As $G'$ contains no non-minimal branches, $Z$ is a count of additional root vertices that are unlabeled, unattached and have two or more children, all of which are either labeled, attached or both. 
Using this information, it is straightforward to demonstrate that there exists a modified pair of history graphs $(H, H')$ such that $H$ has the same size and cardinality as $G$, and $H'$ is a $H$-bounded extension of $H$ that has more labeled or attached vertices than $G'$. The existence of $(H, H')$ contradicts either or both Lemmas \ref{labelsBoundLemma} or Lemma \ref{adjacenciesBoundLemma}.

Figure \ref{sizeOfGBoundedGraph3} shows this bound is tight for all $n$ and $m$.
\end{proof}


\begin{figure}[h!]
\begin{center}
\includegraphics[width=7cm]{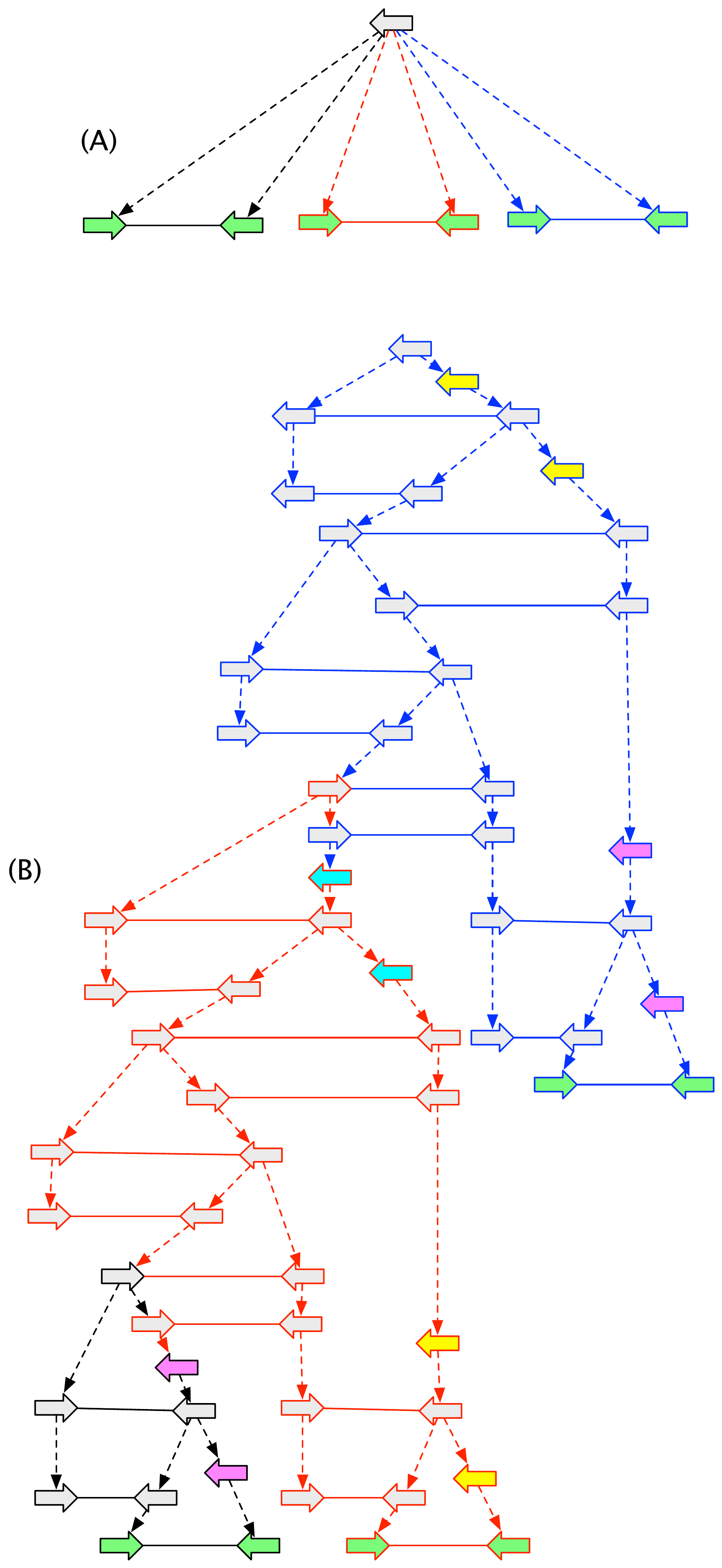}
\caption{\textbf{(A)} The combination of the history graphs in Figures \ref{sizeOfGBoundedGraph1}(A) and  \ref{sizeOfGBoundedGraph2}(B), constructed by merging their root vertices. \textbf{(B)} A $G$-bounded AVG extension of the graph in (B) with 22 extra adjacencies and 54 extra vertices, the maximum possible. The colouring of the elements is used to demonstrate the bound is tight for any combination of $m$ and $n$, and follows that used in Figure \ref{sizeOfGBoundedGraph1}.}
\label{sizeOfGBoundedGraph3}
\end{center}
\end{figure}

\subsection{Appendix D}

In this section we will prove Theorem \ref{posetTheorem}.

A adjacency $\{ x_{\alpha}, y_{\beta} \}$ is \emph{old} if both $A(x_{\alpha})$ and  $A(y_{\beta})$ are each independently either the side of a free-root or incident with a $G$-irreducible adjacency.

\begin{figure}[h!]
\begin{center}
\includegraphics[width=11cm]{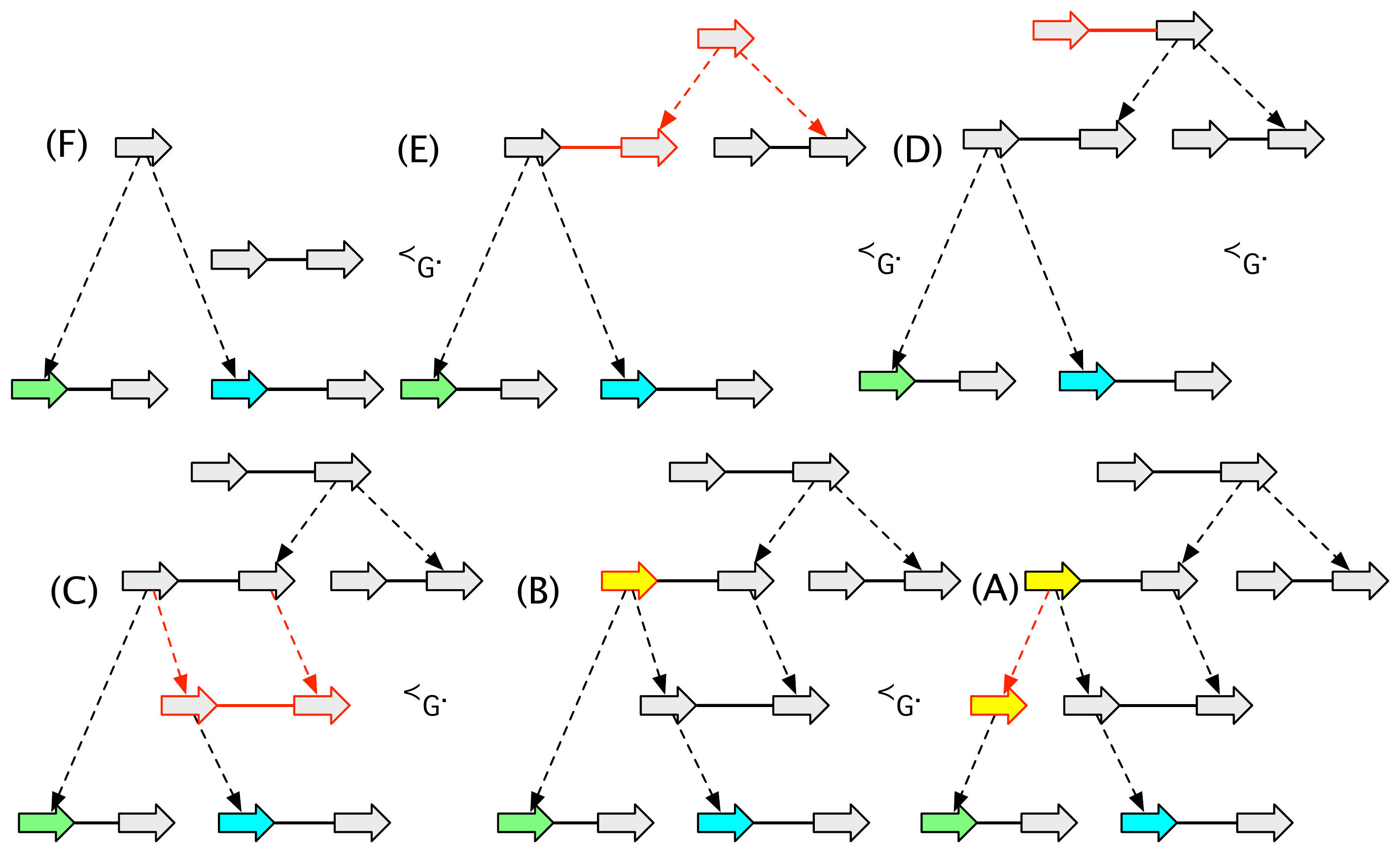}
\caption{A sequence of label/adjacency detachments that transform a $G$-bounded history graph into $G$.  Label detachments of a \textbf{(A-B)} bridge label and \textbf{(B-C)} junction label. Bond detachments of a \textbf{(C-D)} bridge adjacency and (old) \textbf{(D-E-F)} junction adjacencies. Elements outlined in red are those being removed.}
\label{walkingBackwardsExample}
\end{center}
\end{figure}

\begin{lemma}
\label{walkingBackwardsLemma}
For any $G$-bounded history graph $G'$ not isomorphic to $G$ there exists a label detachment or adjacency detachment that results in a $G$-bounded history graph. 
\end{lemma}

\begin{proof}
As $G$ is not isomorphic to $G'$, $G'$ contains one or more $G$-reducible elements. 
If there exists a $G$-reducible label that is a bridge then its label detachment results in a $G$-bounded reduction (see Figure \ref{walkingBackwardsExample}(A-B)). 
Else if there exists a $G$-reducible label it is a junction label and its label detachment results in a $G$-bounded reduction (see Figure \ref{walkingBackwardsExample}(B-C)).
Else if there exists a $G$-reducible adjacency that is a bridge then its adjacency detachment results in a $G$-bounded reduction (see Figure \ref{walkingBackwardsExample}(C-D)). 
Else there exists a $G$-reducible adjacency that is an old junction adjacency and whose adjacency detachment results in a $G$-bounded reduction (see Figure \ref{walkingBackwardsExample}(D-E-F)). 
\end{proof}

The previous lemma implies that for any $G$-bounded history graph there exists a sequence of label and adjacency detachments that results in $G$. We now seek the inverse, to demonstrate the existence of a sequence of moves to create a $G$-bounded AVG from any $G$-bounded history graph.

The inverse of a label/adjacency/lateral-adjacency detachment is, respectively, a \emph{label/adjacency/lateral-adjacency attachment}.

\begin{figure}[h!]
\begin{center}
\includegraphics[width=8cm]{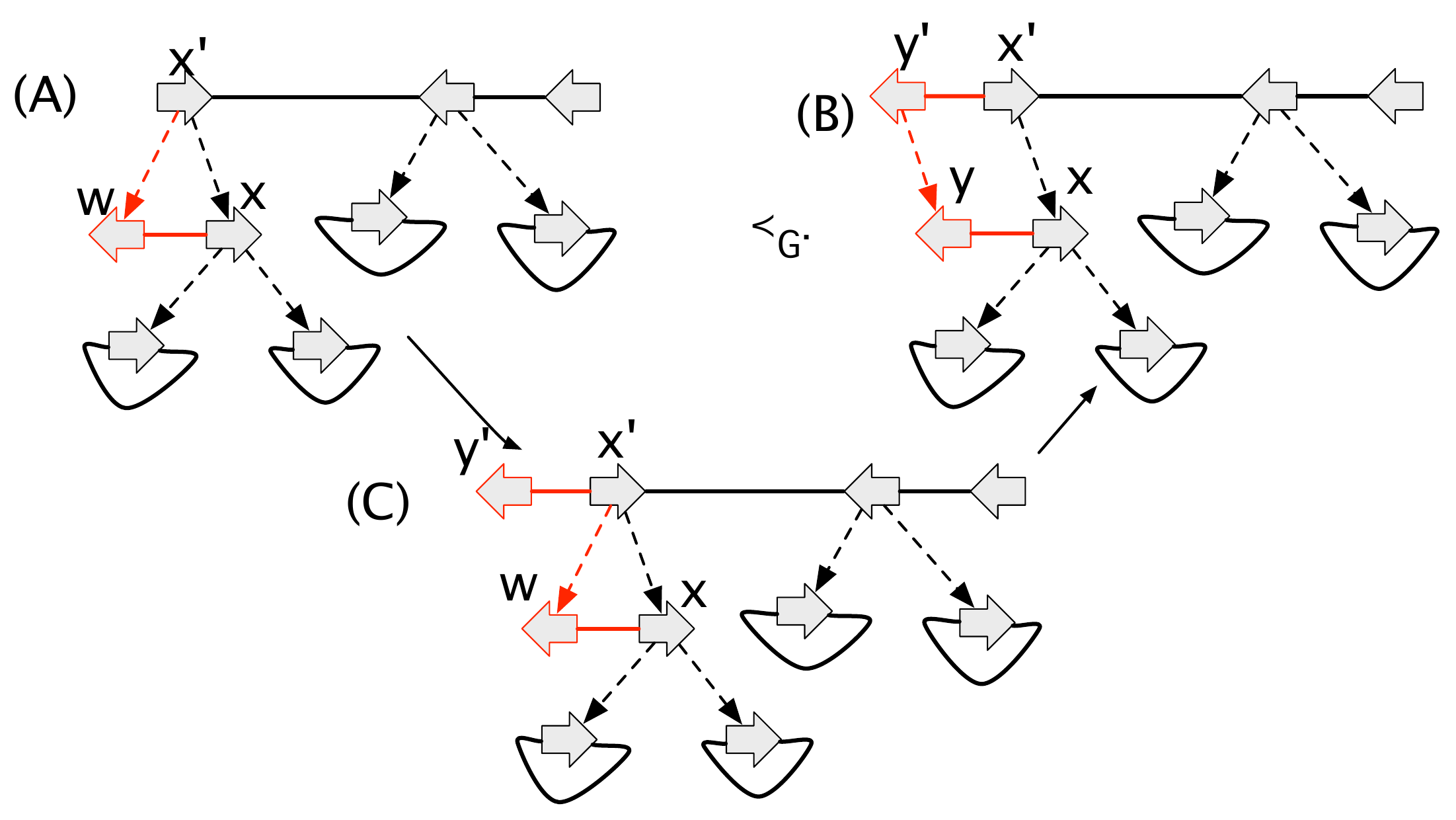}
\caption{\textbf{(A-B)} Lateral-adjacency detachment. Elements in red are $G$-reducible. \textbf{(C)} A not $G$-bounded intermediate of the lateral-adjacency detachment that contains a $G$-reducible ping adjacency.}
\label{lateralBondDetachment}
\end{center}
\end{figure}

The graph in Figure \ref{lateralBondDetachment}(A) has no valid label/adjacency attachment operation that results in a  $G$-bounded history graph, yet it is not an AVG, because it has an unattached junction side, $x'_{tail}$ and any adjacency attachment of $x'_{tail}$ results in the creation of a $G$-reducible ping adjacency. This motivates the need for the lateral-adjacency  detachment/attachment operation, that we use to avoid the creation of $G$-reducible ping adjacencies. 
Notably, while both label detachment and adjacency detachment operations define reductions, the result of a lateral-adjacency detachment, though an extension of $G$, is not necessarily a reduction of the starting graph, though it always has one fewer adjacency. 

\begin{lemma}
\label{walkingForwardsLemma}
A $G$-bounded history graph $G'$ such that $u(G) > 0$ has a label/adjacency/lateral-adjacency attachment that results in a $G$-bounded history graph.
\end{lemma}

\begin{proof}
If $G'$ has a free-root $x$ such that $|L'_{x}| > 1$, then the labeling of the root of the branch-tree whose free-root is $x$ is a label attachment that results in a $G$-bounded extension that contains an additional junction label (see Figure 7(A-B) in the main text).
Else, if $G'$ has substitution ambiguity then there exists a labeled vertex with two or more non-trivial lifted labels for which there exists a label attachment that results in a $G$-bounded extension, which contains an additional bridge label (see Figure 7(B-C) in the main text).
Else $G'$ has rearrangement ambiguity. 
If $G'$ has one or more unattached junction sides, let $x_{\alpha}$ be such a side.
If the most ancestral attached descendants of $x_{\alpha}$ are not incident with hanging adjacencies then the creation of an isolated vertex $y$ and adjacency $\{ x_{\alpha}, y_{\alpha} \}$ is an adjacency attachment that results in a $G$-bounded extension (see Figure 7(C-D) in the main text). 
Else there exists a lateral-adjacency attachment that results in a $G$-bounded history graph in which $x_{\alpha}$ is incident with an adjacency with no hanging endpoints (see Figure 7(D-E), the operation is also an adjacency attachment in this example). 
Else $G'$ does not have an unattached junction side, and there exists an attached junction side with two or more incident non-trivial lifted adjacencies for which there exists an adjacency attachment that results in a $G$-bounded extension that contains an additional bridge adjacency (see Figure 7(E-F) in the main text). 
\end{proof}

Given Theorem \ref{finiteGBoundedTheorem}, the previous lemma implies that for any $G$-bounded history graph there exists a sequence of label/adjacency/lateral-adjacency attachment operations that result in a $G$-bounded AVG.

\newtheorem*{theorem4}{Theorem 4}

\begin{theorem4}
The $G$-bounded poset is finite, has a single least element G, its set of maximal elements are AVGs, and if and only if there exists a $G$-bounded reduction operation to transform $G''$ into $G'$ then $G' \red\cdot_G G''$.
\end{theorem4}

\begin{proof}
That $G$-bounded is finite follows from Theorem \ref{finiteGBoundedTheorem}. 
Lemma \ref{walkingBackwardsLemma} implies it has a single least element.
As a corollary of Theorem \ref{finiteGBoundedTheorem} and Lemma \ref{walkingForwardsLemma} it follows that the set of maximal elements of the $G$-bounded poset are AVGs. 

It remains to prove $G' \red\cdot_G G''$ if and only if there exists a $G$-bounded reduction operation to transform $G''$ into $G'$. The only if follows by definition.
If $G' \red_G G''$ but not $G' \red\cdot_G G''$  then there exists a $G'''$ such that $G' \red_G G''' \red_G G''$. If $G''$ is transformed to $G'$ by a single $G$-bounded reduction operation, to complete the proof it is sufficient to show that no such $G'''$ can exist, this is easily verified. 
\end{proof}

\subsection{Appendix E}

\newtheorem*{lemma9}{Lemma 9}

\begin{lemma9}
There exists a history graph $G$ with an infinite number of $G$-minimal extensions. 
\end{lemma9}

\begin{proof}
We will demonstrate there exists an infinite set of $G$-minimal AVG extensions of the history graph $G$ shown in Figure \ref{infiniteMinimalExtensions}(A). The extensions are composed of the repeating subgraph shown in bold in Figure \ref{infiniteMinimalExtensions}(B) and the terminal elements shown in bold in Figure \ref{infiniteMinimalExtensions}(C) that attach the most ancestral copies of $w$ and $y$.

Consider the AVG extension $H_0$ with zero copies of the repeating subunit and the terminal elements to attach $w^0$ and $y^0$, as in Figure \ref{infiniteMinimalExtensions}(C). 
As $a$, $b$, $c$ and $d$ are labeled but no other vertices are labeled, the adjacencies $\{ a_{head}, b_{head} \}$ and $\{ c_{head}, d_{head} \}$ are $G$-irreducible, because removal of either in any reduction would create a graph that can not then be an extension of $G$.  Given this observation, by definition $w^0$ and $y^0$ or any vertices produced by contracting incident branches of  $w^0$ and $y^0$ must be junctions in any $G$-minimal reduction, and therefore be attached, but by definition of the reduction relation, $\{ w^0_{head}, w^{'0}_{head} \}$ and $\{ y^0_{head}, y^{'0}_{head} \}$ can not be removed and yet $w^0$ and $y^0$ be attached in any $G$-minimal reduction. This therefore implies that the bridge adjacencies $\{ x^0_{head}, x^{'0}_{head} \}$ and $\{ z^0_{head}, z^{'0}_{head} \}$ are also not removed in a $G$-minimal reduction, but this is all the adjacencies in $H_0$, as all the vertices in $H_0$ are attached, therefore $H_0$ is $G$-minimal.

Let $H_i$ be an AVG with $i$ such layers, where $i > 0$ (Figure \ref{infiniteMinimalExtensions}(D) shows an example for $i=2$). To prove that $H_i$ is a $G$-minimal AVG extension we proceed by induction. $H_0$ is the base case. Assume the adjacencies incident with $w^{i-1}$ and $y^{i-1}$ are not removed in any $G$-minimal reduction. Using similar logic to the base case the adjacencies incident $w^i$, $x^i$, $y^i$ and $z^i$ are similarly not removed in a $G$-minimal reduction, again as this is all the added adjacencies and all vertices are attached, using the induction therefore $H_i$ is $G$-minimal.
\end{proof}

\begin{figure}[h!]
\begin{center}
\includegraphics[width=14cm]{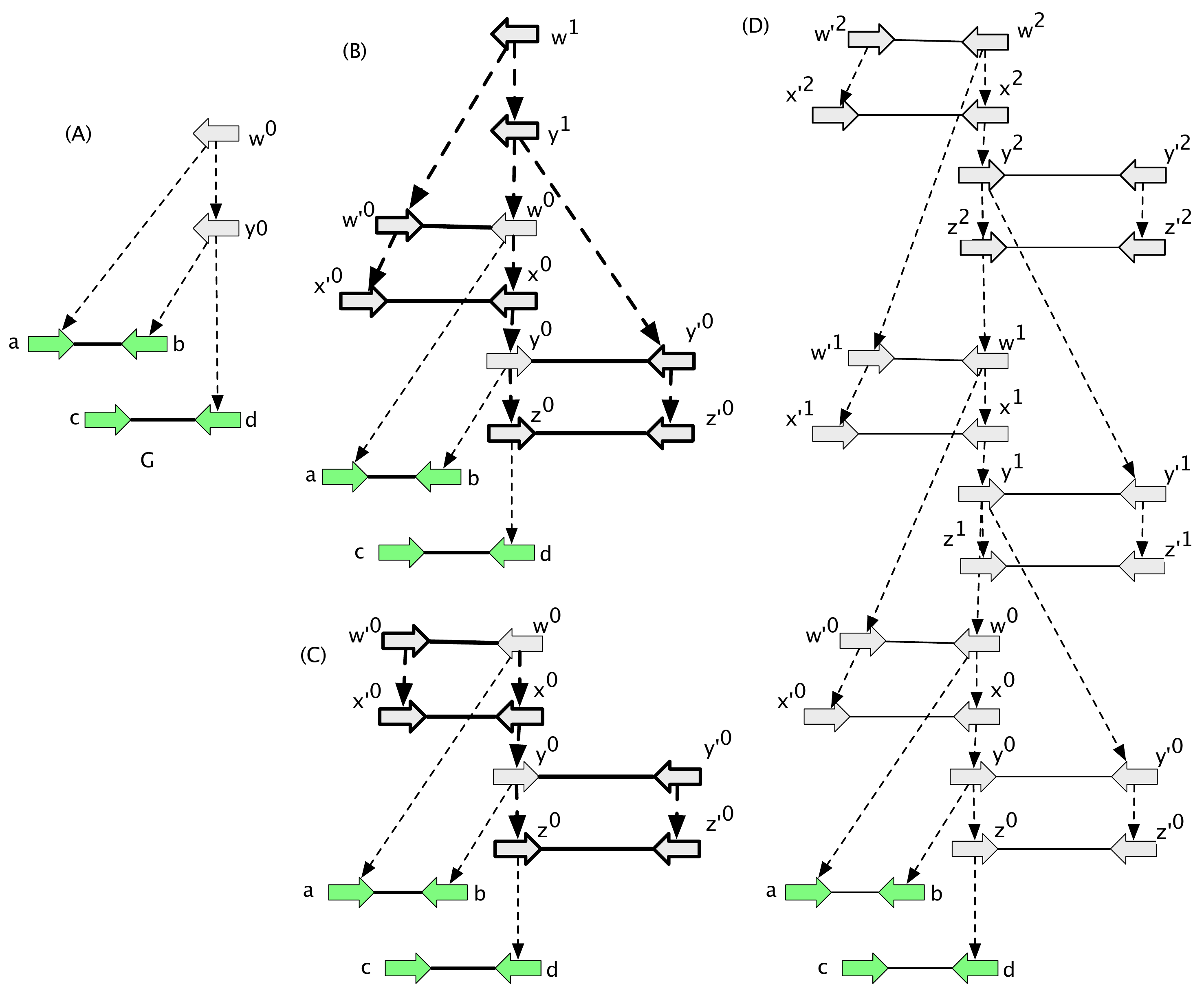}
\caption{An illustration of an infinite set of G-minimal AVG extensions. \emph{(A)} $G$. \emph{(B)} $G$ with an added single copy of the repeating unit in bold. \emph{(C)} $G$ with the added terminal elements in bold. \emph{(D)} A $G$-minimal AVG extension with 2 copies of the repeat subunit and the terminal elements that attach $w^2$ and $y^2$. Vertices are identified by lowercase letters, with superscripts used to denote distinct copies of elements in the repeating subunit.}
\label{infiniteGMinimal}
\end{center}
\end{figure}

\section{Competing interests}
  The authors declare that they have no competing interests.

\section{Author's contributions}
    BP, DZ, GH and DH developed the theory. BP and DZ implemented the theory and performed the experiments. BP wrote the paper, which was edited by DZ, GH and DH.
    
\section{Acknowledgements}
 We would like to thank Dent Earl for his help with figures and the Howard Hughes Medical Institute, Dr. and Mrs. Gordon Ringold, NIH grant 2U41 HG002371-13 and NHGRI/NIH grant 5U01HG004695 for providing funding.


\bibliography{manuscript} 
\bibliographystyle{plain}





\end{document}